\documentclass[mnsc,nonblindrev]{informs3-ssrn}
\RequirePackage{tgtermes}
\RequirePackage{newtxtext}
\RequirePackage{newtxmath}
\RequirePackage{bm}
\RequirePackage{endnotes}

\OneAndAHalfSpacedXI 

\usepackage[utf8]{inputenc}

\usepackage{algorithm}
\usepackage{algpseudocode}
\usepackage{amsmath}
\usepackage{amssymb}
\usepackage{subcaption}
\usepackage{array}
\usepackage{todonotes}
\usepackage{longtable}
\usepackage{tikz}
\usepackage{pgfplots}
\usepackage{titlesec}

\titlespacing*{\section}
{0pt}          
{1.5ex plus 1ex minus .2ex} 
{1.5ex plus .2ex}           

\usepackage{xcolor}

\usepackage{natbib}
 \bibpunct[, ]{(}{)}{,}{a}{}{,}%
 \def\bibfont{\small}%
\let\cite\citep

\EquationsNumberedThrough    

\TheoremsNumberedThrough     
\ECRepeatTheorems  %

\MANUSCRIPTNO{}

\begin{document}

\makeatletter
\renewcommand{\thefootnote}{\fnsymbol{footnote}}
\makeatother
\setcounter{footnote}{1}


\RUNAUTHOR{Umrawal, Stutsman, Woods, and Robson}

\RUNTITLE{How Social Reinforcement Shapes Word-of-Mouth: A Conviction-Based Model for Marketing Diffusion and Network Targeting}


\TITLE{How Social Reinforcement Shapes Word-of-Mouth: A Conviction-Based Model for Marketing Diffusion and Network Targeting}

\ARTICLEAUTHORS{%

\AUTHOR{Abhishek K. Umrawal\footnote{Abhishek K. Umrawal is the corresponding author.}}
\AFF{University of Illinois Urbana-Champaign, Urbana, Illinois 61801, USA \EMAIL{aumrawal@illinois.edu}}

\AUTHOR{Curt Stutsman}
\AFF{University of Illinois Urbana-Champaign, Urbana, Illinois 61801, USA \EMAIL{curtiss8@illinois.edu}}

\AUTHOR{Camille Woods}
\AFF{University of Illinois Urbana-Champaign, Urbana, Illinois 61801, USA \EMAIL{crwoods2@illinois.edu}}

\AUTHOR{Eliot W. Robson 
}
\AFF{Narmi Inc., New York City, New York 10016, USA \EMAIL{eliot.robson24@gmail.com}}
} 

\ABSTRACT{%
Organizations increasingly rely on social networks to drive product adoption and word-of-mouth diffusion, yet a central challenge is that consumers differ not only in whether they adopt but also in how strongly they subsequently advocate. In particular, individuals who adopt under strong social reinforcement often develop greater conviction and become more persuasive. Existing diffusion models largely abstract away from this mechanism by treating influence strength as static.
We introduce a conviction-based diffusion framework, the Pressure Threshold (PT) model, that captures how social reinforcement shapes both adoption and subsequent advocacy. In this model, an individual’s influence increases with the cumulative social pressure experienced at adoption, providing a tractable representation of conviction formation while preserving the interpretability of classical threshold-based models.
We study the managerial implications of this mechanism through marketing targeting and seeding. We show that accounting for reinforcement fundamentally alters optimal targeting strategies, shifting emphasis toward consumers in cohesive network regions where repeated exposure can amplify downstream word-of-mouth. While reinforcement introduces nonstandard features into the optimization problem, we demonstrate that scalable heuristics remain effective and provide data-dependent performance guarantees.
Across synthetic and real-world networks, we find that reinforcement often substantially increases the effectiveness of seeding interventions, particularly in dense and clustered environments where opportunities for social validation are abundant. Using time-stamped social media cascades, we further show that the PT model improves predictive accuracy relative to standard diffusion models, highlighting the empirical relevance of conviction-driven amplification.
}

\maketitle


\section{Introduction}\label{sec:Intro}

Word-of-mouth and social influence play a central role in the diffusion of products, beliefs, and behaviors \cite{bass1969new, goldenberg2001talk, rogers2003diffusion, aral2012influential, libai2013customer}. Consumers are not only influenced by their peers’ adoption decisions, but also by the intensity and consistency of social reinforcement they experience prior to adopting. In many settings, individuals who adopt after repeated exposure within their social network develop stronger conviction and, in turn, become more persuasive advocates to others. Yet most existing approaches implicitly assume that all adopters exert similar levels of influence, regardless of how their adoption was formed. This heterogeneity in post-adoption advocacy is critical for firms and platforms that rely on network-based interventions, such as targeted seeding, influencer marketing, and referral programs, to maximize downstream word-of-mouth and customer lifetime social value \cite{lobel2017customer, libai2013customer}. 

To study these processes, a large literature models social influence using threshold-based or contagion frameworks, in which individuals adopt when the aggregate influence from their neighbors exceeds a threshold \cite{granovetter1978threshold, kempe2003maximizing}. Building on these models, a central decision problem is influence maximization (IM), which seeks to identify a set of initial adopters (or “seed” nodes) that maximizes downstream diffusion \cite{domingos2001mining, kempe2003maximizing}. This framework has been widely used to guide marketing interventions such as targeted seeding, information campaigns, and referral programs.

The modern formalization of network-based diffusion expands on threshold and contagion models of social influence, including the Linear Threshold (LT) and Independent Cascade (IC) frameworks introduced by \citet{kempe2003maximizing}. These models extend earlier theories of collective behavior and social thresholds \citep{granovetter1978threshold} and connect to marketing perspectives on word-of-mouth diffusion and network-driven adoption \citep{goldenberg2001talk, aral2013social, lobel2017customer}. Their appeal lies in their conceptual clarity and tractability, which have made them the dominant models for studying how peer influence propagates through networks and for designing data-driven marketing interventions.

A growing body of empirical and theoretical work highlights the importance of social reinforcement in shaping diffusion outcomes. In many settings, repeated exposure within cohesive or homophilic communities increases both the likelihood of adoption and the intensity of subsequent engagement, consistent with the idea of complex contagion \citep{monsted_plosone2017}. In online environments, users tend to cluster into like-minded groups, and information transmission is disproportionately concentrated within these communities, amplifying beliefs and behaviors through repeated validation \citep{cinelli_pnas2021}. From a marketing perspective, such reinforcement is closely related to mechanisms of social proof and conformity, whereby individuals infer value or credibility from observing widespread adoption among peers. Further marketing and consumer behavior research suggests that the strength and persistence of word-of-mouth are shaped by social context, including peer validation, conformity pressures, and repeated exposure \citep{aral2020hype, pan2025influencer}.

These patterns imply that influence is not merely transmitted across network connections but can be amplified through local reinforcement. In particular, individuals who adopt after observing consistent support from their peers may develop stronger conviction and become more active or persuasive participants in subsequent word-of-mouth processes. This amplification mechanism is central to many real-world diffusion phenomena. Still, it is not explicitly represented in standard diffusion models which assume static influence exerted by an individual, independent of how their own adoption was shaped. As a result, while LT and IC models capture whether individuals adopt in response to social exposure, they do not capture how the process of adoption itself may shape subsequent advocacy behavior. This omission limits their ability to capture reinforcement-driven amplification observed in practice.


To address these limitations, we introduce the Pressure Threshold (PT) model, a tractable extension of the Linear Threshold framework that captures conviction formation through social reinforcement. In the PT model, an individual’s influence is not fixed but instead depends on the cumulative social pressure experienced at the time of adoption. Specifically, individuals who adopt after repeated or consistent exposure to their peers develop greater conviction and, as a result, exert stronger influence on others in subsequent interactions. This formulation provides a simple yet expressive way to operationalize reinforcement-driven amplification within a familiar diffusion framework. By allowing influence strength to evolve endogenously with the adoption process, the model captures a key behavioral mechanism underlying word-of-mouth dynamics while preserving the interpretability and transparency of threshold-based models. As such, the PT model bridges the gap between classical diffusion models and empirical observations of reinforcement, providing a foundation for studying how conviction formation shapes both diffusion patterns and the effectiveness of network-based marketing interventions.


By retaining the additive semantics and interpretability of the Linear Threshold model while making reinforcement explicit, the PT framework offers a principled and transparent way to move beyond static LT and IC assumptions. Importantly, this modeling approach is strongly supported by empirical evidence. As we demonstrate through empirical validation later in the paper, incorporating pressure-based amplification significantly improves the accuracy of diffusion predictions on real-world social networks. The PT model thus provides both a richer behavioral foundation for influence maximization and a practical tool for aligning diffusion models with observed online and offline social dynamics.

\subsection{Contributions}

This paper makes four main contributions that advance the modeling and analysis of social influence under reinforcement, with direct implications for marketing interventions in networks. Our contributions span modeling, theory, and managerial insights.

\begin{enumerate}

    \item \textbf{A Behavioral Model of Conviction-Driven Diffusion.}  
    We introduce the Pressure Threshold (PT) diffusion model, an extension of the Linear Threshold (LT) model that explicitly incorporates local social reinforcement. In PT, a node’s outgoing influence depends on the cumulative social pressure that triggered its activation, capturing the idea that individuals who adopt under stronger peer reinforcement develop greater conviction and become more influential in subsequent interactions. This provides a tractable and interpretable framework for modeling amplification in word-of-mouth processes while preserving the additive structure of threshold-based diffusion.

    \item \textbf{Implications for Targeting and Seeding Decisions.}  
    We study how reinforcement reshapes optimal intervention strategies in the influence maximization problem. We show that amplification fundamentally alters the structure of effective targeting, favoring coordinated seeding in cohesive regions of the network where repeated exposure can strengthen downstream influence. These results provide new insights for marketing strategies such as influencer selection, referral program design, and campaign coordination in networked environments.

    \item \textbf{Empirical Evidence on Amplification and Predictive Performance.}  
    Through experiments on synthetic and real-world networks, as well as validation on timestamped social media cascades, we show that reinforcement-driven diffusion produces substantially different, and consistently larger, cascades than standard LT models, particularly in dense or clustered networks. We further demonstrate that the PT model significantly improves diffusion prediction accuracy, indicating that explicitly modeling social pressure yields a more faithful representation of observed word-of-mouth dynamics.

    \item \textbf{Theoretical and Algorithmic Foundations for Reinforcement-Based Diffusion.}  
    We provide a structural analysis of influence maximization under PT, showing that the problem remains NP-hard and that the influence function is generally neither monotone nor submodular. Despite this departure from classical assumptions, we develop a Sandwich Approximation framework that enables the use of greedy methods with data-dependent guarantees, and demonstrate that scalable heuristics adapted from LT, such as CELF, remain effective in our empirical settings. We also support these analyses with an efficient implementation built on the CyNetDiff \cite{robson2024cynetdiff} library, enabling large-scale simulation and reproducibility.

\end{enumerate}

Together, these results show that incorporating social reinforcement leads to qualitatively different diffusion dynamics and targeting strategies than those predicted by static models, while remaining computationally tractable and practically implementable for network-based decision making.



\subsection{Organization}

The remainder of the paper is organized as follows. 
Section~\ref{sec:conceptual} develops the conceptual motivation and situates the paper within the related literature on reinforcement, persuasion, and amplification in social influence. 
Section~\ref{sec:model} introduces the Pressure Threshold diffusion model, provides intuition for its mechanism, and interprets the reinforcement parameter in behavioral and decision-relevant terms. 
Section~\ref{sec:decision} formulates influence maximization under pressure as a managerial decision problem and analyzes its structural properties, including NP-hardness and the failure of monotonicity and submodularity. 
Section~\ref{sec:model_properties} presents model behavior and develops the sandwich approximation framework, highlighting how greedy methods can be applied effectively despite weakened theoretical guarantees. 
Section~\ref{sec:experiments} reports computational experiments on synthetic and real-world networks, illustrating how reinforcement alters diffusion dynamics, seed selection, and intervention outcomes. 
Section~\ref{sec:greedy} provides an empirical evaluation of the sandwich approximation strategy and demonstrates data-dependent approximation guarantees in practice. 
Section~\ref{sec:validation} evaluates the predictive performance of the PT model using real-world social media cascades. 
Section~\ref{sec:implications} discusses managerial and societal implications for marketing, platform design, and information governance. 
Section~\ref{sec:conclusion} concludes with a summary of contributions and directions for future research.

\vspace{.1in} 
\section{Conceptual Motivation and Related Literature}\label{sec:conceptual}

Diffusion models have a long history across multiple disciplines, including marketing, sociology, and epidemiology, where they are used to understand how behaviors, products, and information spread through populations. In marketing, diffusion has long been studied as a central mechanism underlying product adoption and word-of-mouth dynamics \citep{bass1969new, goldenberg2001talk, rogers2003diffusion, libai2013customer}. More recent work emphasizes the role of social media, network interactions, and digital platforms in shaping diffusion outcomes and firm value \citep{aral2013social, lobel2017customer}. In epidemiology, the Susceptible–Infected–Recovered (SIR) model introduced a compartmental approach that clarified how local transmission mechanisms aggregate into population-level dynamics \citep{SIR}. In the social sciences, threshold-based models provided a complementary behavioral perspective. In particular, \citet{granovetter1978threshold} conceptualized adoption as a function of cumulative peer influence, highlighting how individual decision rules interact with network structure to produce aggregate diffusion patterns.

Building on these foundations, diffusion in networks has been widely modeled using threshold-based and cascade-based frameworks that formalize how peer influence propagates through social structures. Among the most prominent are the Linear Threshold (LT) and Independent Cascade (IC) models introduced by \citet{kempe2003maximizing}, which translate behavioral insights from sociology and marketing into tractable stochastic processes on networks. These models provide a unifying framework for studying how individual-level responses to social exposure generate aggregate diffusion patterns and have been extensively used to analyze word-of-mouth dynamics, social media influence, and network-based marketing interventions \citep{goldenberg2001talk, aral2013social, lobel2017customer}. Their appeal lies not only in analytical tractability but also in their interpretability: both LT and IC explicitly link adoption decisions to peer influence, making them well suited for studying how local interactions shape global outcomes. As a result, they have become the dominant modeling paradigms for diffusion and form the basis for much of the work on targeting, seeding, and intervention design in networked settings.

Within these models, a central decision problem is influence maximization (IM), which seeks to identify a set of initial adopters (or “seed” nodes) that maximizes downstream diffusion. First introduced in the context of viral marketing by \citet{domingos2001mining}, IM provides a formal framework for allocating marketing resources across individuals in a network to maximize word-of-mouth impact. The work of \citet{kempe2003maximizing} established the computational foundations of the problem by proving NP-hardness, identifying conditions under which the influence function is submodular, and introducing greedy approximation algorithms with provable guarantees.

Building on this foundation, a large and active literature has developed algorithms to improve the scalability and practical applicability of IM. These include algorithmic refinements such as Cost-Effective Lazy Forward (CELF) \citep{leskovec2007cost} and its extensions \citep{CELF++}, sampling-based approaches such as Reverse Influence Sampling (RIS) \citep{borgs2014maximizing}, lightweight heuristics such as degree discounting \citep{Degree_Discount}, and more recent advances in scalable implementations \citep{tang2015influence, borgs2014maximizing}. In parallel, researchers have explored extensions that incorporate network structure and learning, including community-aware heuristics \citep{umrawal2023leveraging, umrawal2023community, robson2025community} and adaptive or data-driven approaches \citep{kumar2022influence}.

Beyond the classical formulation, a wide range of variants have been studied to better capture real-world marketing and decision-making settings. These include models with partial or fractional incentives \citep{chen2020scalable,UmrawalAQ2023,umrawal2026maximizing,BhimarajuRVU2024}, online and adaptive influence maximization \citep{mallipeddi2022framework, agarwal2022stochastic,nie2022explore, xu2026lofa}, fairness-aware interventions \citep{becker2022fairness,lin2023fair,nguyen2022fairnessbudget}, and learning-based approximations of influence functions \citep{kumar2022influence}. Comprehensive surveys of these developments can be found in \citet{li2023surveyml}. Collectively, this body of work has significantly advanced the algorithmic treatment of influence maximization while largely retaining the static diffusion assumptions inherited from LT and IC.

In parallel, a separate line of work has sought to move beyond static diffusion models by incorporating temporal dynamics and evolving network structure. These approaches recognize that social systems are often nonstationary, with influence patterns and network connections changing over time. For example, \citet{Zhuang} studied diffusion processes with time-varying edges, while \citet{Xie} proposed DynaDiffuse, which combines continuous-time propagation with dynamically evolving influence weights. More broadly, continuous-time and dynamic network models have been used to capture temporal patterns in information spread and user interaction, particularly in social media and online platforms.

From a marketing perspective, such dynamic models are closely related to efforts to understand how timing, exposure frequency, and campaign scheduling affect consumer behavior and diffusion outcomes \citep{aral2013social, pan2025influencer, kanuri2018scheduling}. These models provide important insights into when and how influence propagates, particularly in settings where user activity and connections evolve. However, they primarily focus on temporal and structural variation in diffusion processes, rather than on how the experience of social reinforcement influences individual behavior.

In contrast, the Pressure Threshold model proposed in this paper focuses on a complementary dimension: how cumulative social pressure at the point of adoption shapes subsequent influence. While dynamic diffusion models relax assumptions about time and network structure, they generally retain a static notion of influence strength conditional on adoption. The PT model instead endogenizes influence strength through reinforcement, enabling the study of amplification effects that arise from repeated exposure and peer validation.

Several studies have examined the empirical accuracy of diffusion models by comparing simulated cascades with real-world diffusion data. \citet{Kuo} evaluated predictive performance across models and found that IC often performs better for direct prediction tasks, whereas LT exhibits advantages in indirect or aggregated comparisons. \citet{Ananthasubramaniam2024TheRO} analyzed Twitter hashtag propagation and showed that predictive accuracy depends strongly on both network topology and content characteristics, highlighting the sensitivity of diffusion outcomes to modeling assumptions.

From a marketing perspective, these findings underscore the importance of selecting diffusion models that accurately reflect underlying behavioral mechanisms, particularly in settings where firms rely on predictive models to guide targeting and campaign design. Empirical work by \citet{aral2012influential} provides direct evidence of cumulative social influence in product adoption, supporting threshold-based representations of peer effects. More recent work in influencer marketing further demonstrates the importance of repeated exposure and reinforcement in driving engagement and diffusion outcomes. For example, \citet{mallipeddi2022framework} show using large-scale social media data that multiple exposures lead to non-linear, diminishing returns in engagement and that message timing and spacing significantly shape diffusion outcomes. More broadly, research on social media and influencer marketing emphasizes that exposure, repetition, and peer validation play a central role in shaping engagement and word-of-mouth dynamics \citep{aral2013social, pan2025influencer}.

Taken together, this evidence suggests that the predictive performance of diffusion models depends critically on how they capture user behavior and social influence. While existing models account for exposure and network structure, they typically do not incorporate reinforcement effects that arise from repeated or consistent peer influence. This limitation is particularly important in modern online and marketing contexts, where repeated exposure and peer validation are central to both adoption and subsequent influence.

The Pressure Threshold model introduced in this paper addresses this gap by incorporating local social reinforcement directly into the diffusion process. By allowing influence strength to depend on cumulative social pressure at the time of activation, the PT framework provides a tractable and interpretable way to model conviction formation and amplification, while remaining compatible with influence maximization analysis and empirical validation.

\vspace{.1in} \section{Pressure Threshold Diffusion Model}
\label{sec:model}

This section introduces the Pressure Threshold (PT) diffusion model and formalizes the mechanism by which local social reinforcement shapes influence propagation. Our goal is to provide a transparent and behaviorally interpretable model that captures how conviction formed through peer pressure can amplify downstream influence, while remaining closely connected to classical threshold-based diffusion frameworks. This mechanism is motivated by empirical evidence from social media and marketing contexts, where repeated exposure, peer validation, and cumulative persuasion are known to shape both adoption decisions and subsequent word-of-mouth influence. In particular, recent studies of online engagement suggest that individuals exposed to consistent and reinforcing social signals are not only more likely to adopt but also more likely to actively propagate information.

We begin by reviewing the Linear Threshold (LT) model as a baseline, since PT generalizes its structure and semantics. We then define the PT model, describe its diffusion dynamics in detail, and interpret its key parameters in terms relevant for social and managerial decision contexts. Throughout, we emphasize interpretability and behavioral plausibility, positioning PT as a mechanism for studying reinforcement-driven diffusion prior to any consideration of optimization or algorithmic design.

The Linear Threshold (LT) model occupies a central role in the study of influence diffusion because it provides a parsimonious abstraction of cumulative social influence. In the LT framework, individuals adopt a behavior once the total influence exerted by their neighbors exceeds an individual-specific threshold, reflecting classic theories of collective behavior and social pressure \citep{granovetter1978threshold}. This additive formulation has proven valuable both for its interpretability and for its analytical properties, forming the basis for much of the modern influence maximization literature \citep{kempe2003maximizing}.

Beyond its conceptual clarity, the LT model has become a standard reference point for algorithmic analysis because it cleanly separates network structure from behavioral response. Influence is encoded through fixed edge weights, while heterogeneity in adoption behavior is captured through exogenous thresholds. This structure enables rigorous study of diffusion dynamics and intervention strategies, and it underpins many of the approximation guarantees developed for influence maximization.

At the same time, the LT model treats influence strength as static and does not account for the possibility that individuals who adopt under strong peer pressure may subsequently become more influential advocates. As discussed in Section~\ref{sec:conceptual}, empirical evidence from online and offline social systems suggests that local reinforcement and social validation can substantially amplify downstream influence. These patterns are especially salient in modern digital environments, where repeated exposure and coordinated messaging can strengthen conviction and increase the likelihood of subsequent word-of-mouth propagation. These observations motivate extensions of the LT framework that preserve its interpretability while relaxing its static influence assumption.

We next present the formal definition of the Linear Threshold model, which establishes notation and serves as the baseline against which the Pressure Threshold model is defined.

\subsection{Formal Definition: Linear Threshold Model}

Given that the Pressure Threshold model generalizes the Linear Threshold (LT) model, we begin by formally defining the LT diffusion process. Let $G = (V,E)$ be a directed graph, where $V$ is the set of nodes and $E \subseteq V \times V$ is the set of directed edges. Each edge $(u,v) \in E$ is assigned a weight $w_{uv} \in [0,1]$, representing the influence exerted by node $u$ on node $v$. The weights satisfy the normalization condition
\[
\sum_{u \in N(v)} w_{uv} \leq 1 \quad \text{for each } v \in V,
\]
where $N(v)$ denotes the set of in-neighbors of node $v$.

Each node $v \in V$ is also assigned an activation threshold $\theta_v$, drawn independently and uniformly at random from the interval $(0,1]$. Intuitively, $\theta_v$ represents the minimum cumulative influence required for node $v$ to adopt the behavior.

Diffusion proceeds in discrete time steps. Let $S \subseteq V$ denote an initial seed set of active nodes, and let $A_t \subseteq V$ denote the set of active nodes at time $t$, with $A_0 = S$. At each subsequent time step, an inactive node $v \in V \setminus A_t$ becomes active if the cumulative influence from its active in-neighbors meets or exceeds its threshold:
\[
\sum_{u \in P(v)} w_{uv} \geq \theta_v,
\]
where
\[
P(v) = \{\, u \in A_t : (u,v) \in E \,\}
\]
denotes the set of currently active in-neighbors of $v$. Once a node becomes active, it remains active permanently. The diffusion process terminates when no additional nodes become active.

For clarity, we consistently distinguish between the static set of in-neighbors $N(v)$ and the time-dependent subset $P(v)$ of in-neighbors that are active at time $t$.

We define the influence of a seed set $S \subseteq V$ under the Linear Threshold model as the expected number of nodes activated by the diffusion process initiated from $S$, where the expectation is taken over the random thresholds $\theta_v \in (0,1]$. More generally, for any diffusion model $M$, we denote the expected influence of $S$ by $\sigma_M(S)$.

The influence maximization problem under a diffusion model $M$ is then defined as follows. Given a directed graph $G = (V,E)$ and a budget $k$, select a seed set $S \subseteq V$ with $|S| \le k$ that maximizes expected influence:
\[
\arg\max_{S \subseteq V} \ \sigma_M(S)
\quad \text{subject to} \quad |S| \le k.
\]
This formulation provides a direct connection between diffusion dynamics and strategic intervention, and it serves as the decision problem studied throughout the remainder of the paper.

While the Linear Threshold model provides a clean and interpretable framework for modeling cumulative social influence, it assumes that influence strength is static throughout the diffusion process. In particular, once a node becomes active, the influence it exerts on its neighbors is fixed by the underlying network structure and does not depend on the circumstances under which the node itself was adopted.

A growing body of empirical evidence suggests that this assumption may be overly restrictive in many social systems. Individuals who adopt a behavior after experiencing strong peer pressure or broad social support often develop greater conviction and become more persuasive advocates in turn. These effects are especially pronounced in modern digital and marketing environments, where repeated exposure, coordinated messaging, and peer validation can significantly amplify downstream influence. Static diffusion models such as LT do not capture this feedback between the pressure that triggers adoption and the influence exerted thereafter.

The Pressure Threshold (PT) model relaxes this static influence assumption while preserving the core structure and interpretability of the LT framework. Rather than modifying the activation rule itself, PT introduces a localized reinforcement mechanism that adjusts a node’s outgoing influence at the moment it becomes active, based on the cumulative social pressure it experienced. This design ensures that reinforcement arises as a consequence of adoption, rather than as an independent driver of activation.

By making reinforcement explicit but minimal, the PT model maintains the additive semantics of threshold-based diffusion, avoids recursive feedback loops, and remains compatible with efficient simulation and empirical calibration. We next present the formal definition of the PT diffusion process, beginning with the shared activation rule and then specifying how influence is updated through a two-phase mechanism.

While the activation rule in the Pressure Threshold model is identical to that of the Linear Threshold model, PT departs from LT in how influence evolves after a node becomes active. To capture reinforcement in a controlled and interpretable manner, the PT diffusion process is organized into two distinct phases at each time step.

\paragraph{Two-phase diffusion structure.}
Each round of the PT diffusion process consists of the following phases:

\begin{enumerate}
    \item \textbf{Activation Phase.}  
    All inactive nodes evaluate their cumulative incoming influence based on the current edge weights. Nodes whose cumulative influence meets or exceeds their threshold become newly active, exactly as in the LT model.

    \item \textbf{Influence Adjustment Phase.}  
    Nodes that activate during the current round update the influence they exert on their neighbors as a function of the social pressure they experienced at activation.
\end{enumerate}

This separation emphasizes that reinforcement in PT is a consequence of adoption rather than a driver of it. Influence is first evaluated using the existing network weights, and only after adoption does the model allow a node’s outgoing influence to change. By design, this prevents reinforcement from affecting activation decisions within the same time step and avoids recursive feedback effects that could otherwise complicate diffusion dynamics.

\paragraph{Incoming pressure at activation.}
To formalize reinforcement, we first define the total incoming influence received by a node at the moment it becomes active. Let $v$ be a node that activates at time $t$. We define the \emph{incoming pressure} experienced by $v$ at activation as
\[
I_v = \sum_{u \in P(v)} w_{uv},
\]
where $P(v) = \{\, u \in A_t : (u,v) \in E \,\}$ denotes the set of active in-neighbors of $v$ at the time of activation.

The quantity $I_v$ summarizes the cumulative social influence that triggered node $v$’s adoption. In the PT model, this pressure plays a central role in determining how strongly $v$ influences others after becoming active. Importantly, $I_v$ is evaluated using the pre-activation edge weights and is computed only once, at the time of activation.

\subsection{Formal Definition: Pressure Threshold Diffusion Process}

We now formally define the Pressure Threshold (PT) diffusion model. The PT model is defined on the same underlying graph structure as the Linear Threshold model and retains the same threshold-based activation rule. As a result, PT preserves the core semantics of cumulative influence while modifying how influence evolves after activation through a reinforcement mechanism tied to experienced social pressure.

Let $G = (V,E)$ be a directed graph, where $V$ denotes the set of nodes and $E \subseteq V \times V$ denotes the set of directed edges. Each edge $(u,v) \in E$ is assigned a weight $w_{uv} \in [0,1]$ such that
\[
\sum_{u \in N(v)} w_{uv} \leq 1 \quad \text{for each } v \in V,
\]
where $N(v)$ denotes the set of in-neighbors of $v$. Unless otherwise stated, we initialize weights according to the weighted-cascade convention, $w_{uv} = 1/\text{in-degree}(v)$. Each node $v \in V$ is assigned an activation threshold $\theta_v \in (0,1]$, drawn independently and uniformly at random.

Diffusion unfolds in discrete time steps and begins from an initial seed set $S \subseteq V$, which is assumed to be active at time $t=0$. Let $A_t \subseteq V$ denote the set of active nodes at time $t$, with $A_0 = S$. At each subsequent time step, inactive nodes evaluate whether the cumulative influence from their active in-neighbors meets or exceeds their threshold.

Formally, at time step $t+1$, a node $v \in V \setminus A_t$ becomes active if
\[
\sum_{u \in P(v)} w_{uv} \geq \theta_v,
\]
where $P(v) = \{\, u \in A_t : (u,v) \in E \,\}$ denotes the set of currently active in-neighbors of $v$. As in the LT model, the PT diffusion process is progressive: once a node becomes active, it remains active permanently. Diffusion continues until no additional nodes activate.

\paragraph{Influence amplification.}
The PT model translates incoming pressure at activation into increased outgoing influence through a simple additive update rule. Let $v$ be a node that becomes active at time $t$, and let $I_v$ denote the total incoming influence received at activation, as defined above. The strength of reinforcement is governed by a nonnegative amplification parameter $\alpha \geq 0$. This parameter captures the degree to which individuals internalize social reinforcement when forming persuasive influence.

For each outgoing edge $(v,s) \in E$ such that the neighbor $s$ is inactive at time $t$, the influence weight from $v$ to $s$ is updated according to
\[
w'_{vs} = \min \bigl( 1,\; w_{vs} + \alpha \cdot I_v \bigr),
\]
where $w'_{vs}$ denotes the updated influence weight after reinforcement.

The use of the $\min(1,\cdot)$ operator ensures that all influence weights remain within the unit interval, consistent with the range of activation thresholds. Because the amplification is additive and proportional to the pressure that triggered activation, nodes that adopt after experiencing greater cumulative influence exert correspondingly stronger influence on their neighbors.

\paragraph{One-shot and localized reinforcement.}
Reinforcement in the PT model is applied only once, at the moment a node becomes active. After this update, the outgoing weights of node $v$ remain fixed for the remainder of the diffusion process. Moreover, influence adjustments are applied only to edges directed toward currently inactive neighbors; edges pointing to already active nodes are not updated, since additional influence along those edges would not affect future diffusion.

These restrictions are deliberate. By limiting reinforcement to a single update and restricting it to edges that can influence future activations, the PT model avoids recursive amplification loops and maintains stable diffusion dynamics even in dense networks. At the same time, the additive update rule preserves the interpretability and cumulative nature of threshold-based influence aggregation.

\paragraph{Influence under the PT model.}
As in the Linear Threshold model, the diffusion process under PT is progressive: nodes may transition from inactive to active, but never the reverse. Because reinforcement is applied only once at the moment of activation and no subsequent feedback occurs, the PT diffusion process is guaranteed to terminate in a finite number of steps.

We define the influence of a seed set $S \subseteq V$ under the PT model as the expected number of nodes activated by the diffusion process initiated from $S$, where the expectation is taken over the random thresholds. Consistent with our earlier notation, we denote this quantity by $\sigma_{\mathrm{PT}}(S)$. This influence measure serves as the primary performance metric for both empirical evaluation and the influence maximization decision problem analyzed in subsequent sections.

\paragraph{Interpretation of the reinforcement parameter.}
The amplification parameter $\alpha$ governs the strength of social reinforcement in the PT model. When $\alpha = 0$, the influence adjustment phase is inactive, and the PT model reduces exactly to the classical Linear Threshold model. Small values of $\alpha$ correspond to diffusion dynamics that remain close to LT, while larger values induce stronger reinforcement effects, allowing nodes that adopt under high pressure to exert substantially greater downstream influence.

Importantly, $\alpha$ modulates behavioral feedback rather than network structure. As such, it can be interpreted as a domain-dependent reinforcement coefficient reflecting how strongly individuals internalize social validation when forming convictions. In practice, $\alpha$ can be treated as a tunable hyperparameter and calibrated using empirical diffusion traces, analogous to the way contagion strength or persuasion intensity is estimated in epidemiological and social influence models.

\paragraph{Algorithmic summary.}
For clarity and reproducibility, Algorithm~\ref{alg:pressure-threshold} summarizes the complete Pressure Threshold diffusion process, including both the activation and influence adjustment phases. The algorithm provides an explicit procedural description of PT diffusion and will serve as the basis for all simulations and experiments reported later in the paper.

\begin{algorithm}[tb]
\caption{Pressure Threshold (PT) Diffusion}
\label{alg:pressure-threshold}
\begin{algorithmic}[1]
\Require Directed graph $G=(V,E)$; influence weights $w_{uv} \in [0,1]$; thresholds $\theta_v \in (0,1]$; seed set $S \subseteq V$; amplification parameter $\alpha \ge 0$
\State $t \gets 0$; $A_t \gets S$ \Comment{initial active set}
\Repeat
    \State $N_t \gets \emptyset$ \Comment{nodes newly activating at time $t$}
    \ForAll{$v \in V \setminus A_t$} \Comment{activation phase}
        \State $P_t(v) \gets \{\, u \in A_t : (u,v)\in E \,\}$
        \State $I_v \gets \sum_{u \in P_t(v)} w_{uv}$
        \If{$I_v \ge \theta_v$}
            \State $N_t \gets N_t \cup \{v\}$
        \EndIf
    \EndFor
    \If{$N_t = \emptyset$}
        \State \textbf{break} \Comment{no new activations; diffusion halts}
    \EndIf
    \ForAll{$v \in N_t$} \Comment{influence adjustment phase (one-shot)}
        \ForAll{$(v,s) \in E$ \textbf{with} $s \notin A_t$}
            \State $w_{vs} \gets \min \bigl(1,\; w_{vs} + \alpha \cdot I_v \bigr)$
        \EndFor
    \EndFor
    \State $A_{t+1} \gets A_t \cup N_t$; $t \gets t+1$
\Until{false}
\State \Return $A_t$
\end{algorithmic}
\end{algorithm}

Together, the definitions above fully specify the Pressure Threshold diffusion model. PT retains the cumulative activation semantics and interpretability of the Linear Threshold framework while introducing a controlled form of social reinforcement through a one-shot amplification mechanism. By tying outgoing influence to the pressure experienced at adoption, the model captures conviction-driven amplification effects observed in real social systems without introducing recursive feedback or destabilizing dynamics.

With the diffusion process now fully defined, we turn to its implications for decision-making on networks. In particular, we study how social reinforcement alters the influence maximization problem, a canonical task in marketing, information campaigns, and public interventions. The next section formulates influence maximization under the Pressure Threshold model and analyzes its computational and structural properties.

\vspace{.1in} 
\section{Influence Maximization under the Pressure Threshold Model: Implications for Viral Marketing and Network Interventions}
\label{sec:decision}

We now study the implications of the Pressure Threshold (PT) model for decision-making on networks, focusing on the influence maximization (IM) problem. In many marketing and information campaign settings, firms must identify a limited set of individuals to target in order to maximize downstream adoption and word-of-mouth spread. Under classical diffusion models such as Linear Threshold, this problem is well-understood and admits strong structural properties that enable efficient optimization. The PT model, however, fundamentally alters this problem by introducing reinforcement-driven changes in influence strength, making the effectiveness of a seed set depend not only on who adopts, but also on how adoption occurs.

Under classical diffusion models such as the Linear Threshold model, influence maximization is well understood as a combinatorial optimization problem with favorable structural properties. In particular, the expected influence function is known to be monotone and submodular, enabling greedy algorithms with strong approximation guarantees. These guarantees underpin much of the practical appeal of influence maximization methods.

The presence of social reinforcement fundamentally alters this decision landscape. Because influence strength in the PT model depends on the social pressure experienced at activation, the marginal value of selecting a particular node as a seed can depend on complex interactions involving activation order, neighborhood composition, and reinforcement dynamics. As a result, intervention strategies that are optimal under static diffusion assumptions need not perform well when reinforcement effects are present.

We begin by formally stating the influence maximization problem induced by the Pressure Threshold diffusion process. Let $G=(V,E)$ be a directed graph, and let $\sigma_{\mathrm{PT}}(S)$ denote the expected influence of a seed set $S \subseteq V$ under the PT model, as defined in Section~\ref{sec:model}. Given a budget constraint $k \in \mathbb{N}$, the influence maximization problem under PT asks for a seed set that maximizes expected spread:
\[
\max_{S \subseteq V} \ \sigma_{\mathrm{PT}}(S)
\quad \text{subject to} \quad |S| \le k.
\]

This formulation mirrors the classical influence maximization problem studied under the Linear Threshold and Independent Cascade models, differing only in the underlying diffusion dynamics that determine the influence function. However, because PT endogenizes influence strength through reinforcement, the resulting objective function depends on both the set of activated nodes and the pathways through which activation occurs. This distinction fundamentally alters the structural properties that enable efficient optimization under classical models.

Before turning to algorithmic approaches, it is important to understand the inherent computational complexity of influence maximization under the Pressure Threshold model. Establishing such limits clarifies which performance guarantees are achievable in principle and which must instead rely on heuristic or data-driven methods.

In classical diffusion settings, influence maximization is known to be computationally intractable. Under both the Linear Threshold and Independent Cascade models, the problem is NP-hard, even to approximate beyond a constant factor \citep{kempe2003maximizing}. These results motivate the widespread use of greedy heuristics supported by approximation guarantees under favorable structural conditions.

At first glance, the additional structure introduced by reinforcement might appear to simplify the decision problem. Because PT ties influence strength to observed social pressure, one might hope that this feedback introduces regularities that reduce combinatorial complexity or enable more efficient optimization. Whether this intuition holds is not obvious a priori and requires careful analysis.

We next show that this is not the case. Despite its additional behavioral structure, influence maximization under the PT model remains computationally hard. In particular, reinforcement preserves the fundamental intractability of the classical problem, even in regimes where PT reduces to familiar diffusion dynamics. At the same time, reinforcement disrupts key structural properties such as submodularity that underpin classical approximation guarantees, creating new challenges for algorithm design and analysis.

\subsection{Computational Complexity}

We begin by establishing the computational hardness of influence maximization under the Pressure Threshold model.

\begin{theorem}\label{thm:pt-np-hard}
Influence maximization under the Pressure Threshold (PT) diffusion model is NP-hard.
\end{theorem}

\proof{Proof.}
The influence maximization problem under the Linear Threshold (LT) model is known to be NP-hard \citep{kempe2003maximizing}. We reduce this problem to influence maximization under the PT model.

Consider an arbitrary LT instance defined on a graph $G=(V,E)$ with seed budget $k$. Construct the corresponding PT instance on the same graph with the same budget and with amplification parameter $\alpha=0$. When $\alpha=0$, the influence adjustment phase in the PT diffusion process is inactive, and the PT model reduces exactly to the LT model. In particular, for every seed set $S \subseteq V$,
\[
\sigma_{\mathrm{PT}}(S;\alpha=0) = \sigma_{\mathrm{LT}}(S).
\]

Thus, an optimal solution to the PT influence maximization problem with $\alpha=0$ is also an optimal solution to the original LT influence maximization problem. Because this transformation can be performed in polynomial time, a polynomial-time algorithm for PT influence maximization would imply one for LT influence maximization. Therefore, influence maximization under the PT model is NP-hard.
\endproof

The NP-hardness result above establishes that social reinforcement does not alleviate the fundamental computational difficulty of influence maximization. Even though the PT model introduces additional behavioral structure by linking influence strength to social pressure, this structure does not eliminate the combinatorial difficulty inherent in selecting an optimal seed set. As a result, any algorithmic approach to influence maximization under PT must contend with the same worst-case intractability that arises under classical diffusion models.

At the same time, NP-hardness alone does not fully characterize the practical tractability of the problem or the effectiveness of solution methods. In classical settings, influence maximization remains NP-hard despite the availability of efficient greedy algorithms with strong approximation guarantees, driven by favorable structural properties of the influence function.

These guarantees depend critically on properties such as monotonicity and submodularity. Whether such properties continue to hold under the Pressure Threshold model is therefore a central question. Because reinforcement introduces context dependence into diffusion dynamics, it may alter the behavior of the influence function in ways that are incompatible with classical approximation theory.

We next examine how social reinforcement affects the monotonicity and submodularity of the PT influence function, and discuss the implications for optimization in marketing and network intervention settings.

\subsection{Monotonicity Analysis}

We begin by examining the monotonicity of the influence function under the PT model. Monotonicity is a central property in classical influence maximization, as it ensures that adding additional seed nodes cannot reduce expected influence and supports stable targeting strategies in marketing and intervention contexts.

\begin{definition}[Monotonicity]
A set function $f : 2^V \to \mathbb{R}$ is \emph{monotone} if for all seed sets $S \subseteq T \subseteq V$, it holds that
\[
f(S) \le f(T).
\]
\end{definition}

Under the Linear Threshold and Independent Cascade models, the expected influence function is monotone. This property ensures that selecting additional seeds is never harmful and, together with submodularity, enables greedy algorithms with strong approximation guarantees. In the absence of monotonicity, however, standard greedy reasoning no longer applies: adding a seed may reduce overall influence, complicating both analysis and intuition.

We now show that the influence function under the Pressure Threshold model does not satisfy monotonicity in general. Unlike classical diffusion models, monotonicity need not hold under PT due to the endogenous role of reinforcement at activation, which creates a subtle trade-off between early seeding and influence amplification.

\begin{theorem}\label{thm:pt-not-monotone}
There exists a directed graph $G$, a seed budget $k$, and an amplification parameter $\alpha > 0$ such that the PT influence function $\sigma_{\mathrm{PT}}(\cdot)$ is not monotone.
\end{theorem}

\proof{Proof.}
The proof proceeds by construction. We exhibit an instance in which adding a node to the seed set suppresses downstream reinforcement that would otherwise arise through propagation.

Under the PT model, only nodes that activate through diffusion (as opposed to being seeded) receive influence amplification. Thus, seeded nodes do not accumulate incoming pressure and cannot benefit from reinforcement, whereas nodes activated endogenously may do so. As a result, seeding a node directly can prevent that node from activating via social pressure, thereby eliminating the reinforcement effect it would have generated had it been activated endogenously. In such cases, adding a node to the seed set can reduce the total amplified influence in later rounds, leading to a strictly smaller expected cascade.

An explicit example demonstrating this effect is presented in the following subsection.
\endproof

\paragraph{Illustrative example.}
We illustrate the non-monotonicity result in Theorem~\ref{thm:pt-not-monotone} with a simple example that highlights the interaction between seeding and reinforcement.

Consider a directed graph $G=(V,E)$ with nodes arranged so that node $v$ can either be activated directly as a seed or activated indirectly via social pressure from its neighbors. Let the amplification parameter satisfy $\alpha > 0$. Suppose that when $v$ is activated through diffusion, it receives sufficient incoming pressure to substantially amplify its outgoing influence, enabling it to activate a downstream node $u$. In contrast, when $v$ is included in the seed set, it activates immediately but does not receive any incoming pressure and therefore exerts only its baseline influence.

Let $S$ be a seed set that does not include $v$. Under PT diffusion from $S$, node $v$ activates endogenously after accumulating influence from its neighbors, receives reinforcement upon activation, and subsequently activates node $u$. Now consider the augmented seed set $T = S \cup \{v\}$. In this case, node $v$ is activated exogenously at the start of the diffusion process, receives no amplification, and fails to activate $u$. As a result, the final cascade originating from $T$ is strictly smaller than that originating from $S$.

This example demonstrates that under PT diffusion, adding a node to the seed set can suppress downstream reinforcement that would otherwise arise through social pressure. Consequently, the expected influence under PT need not increase monotonically with the size of the seed set.

\subsection{Submodularity Analysis}

We next examine submodularity of the influence function. Submodularity is a diminishing-returns property that plays a central role in the algorithmic theory of influence maximization. Intuitively, a submodular influence function guarantees that the marginal benefit of adding a node as a seed decreases as the seed set grows.

\begin{definition}[Submodularity]
A set function $f : 2^V \to \mathbb{R}$ is \emph{submodular} if for all seed sets $S \subseteq T \subseteq V$ and all nodes $v \in V \setminus T$, it holds that
\[
f(S \cup \{v\}) - f(S) \;\ge\; f(T \cup \{v\}) - f(T).
\]
\end{definition}

Under classical diffusion models such as Linear Threshold and Independent Cascade, the expected influence function is both monotone and submodular. These properties together enable the celebrated $(1 - 1/e)$ approximation guarantee for greedy algorithms \citep{nemhauser1978analysis,kempe2003maximizing}. Submodularity is particularly crucial, as it ensures that greedy selection remains near-optimal despite the combinatorial nature of the problem.

In particular, reinforcement induces strategic complementarities among seed nodes, as the presence of one seed can increase the marginal value of another through amplified propagation. The reinforcement mechanism introduced by the PT model disrupts this diminishing-returns structure. Because activation under PT can amplify future influence, the marginal contribution of a node can increase when it is added later rather than earlier. This context dependence creates complementarities among seeds that are incompatible with classical submodularity.

We formalize this observation below.

\begin{theorem}\label{thm:pt-not-submodular}
There exists a directed graph $G$, a seed budget $k$, and an amplification parameter $\alpha > 0$ such that the PT influence function $\sigma_{\mathrm{PT}}(\cdot)$ is not submodular.
\end{theorem}

\proof{Proof.}
The proof again proceeds by construction. Under PT diffusion, adding a node to a larger seed set can trigger reinforcement effects that do not arise when the same node is added to a smaller seed set. In such cases, the marginal gain from adding a node increases with the size of the seed set, violating the diminishing-returns property required for submodularity.

An explicit construction illustrating this effect is provided in the following subsection.
\endproof

\paragraph{Illustrative example.}
We illustrate the non-submodularity result in Theorem~\ref{thm:pt-not-submodular} with a simple example showing how reinforcement can increase marginal returns as the seed set grows.

Consider a directed graph in which node $v$ can be activated either weakly, with little incoming pressure, or strongly, with substantial incoming pressure supplied by multiple active neighbors. Let $\alpha > 0$. Suppose that when $v$ activates under strong pressure, the resulting amplification of its outgoing influence is sufficient to activate a downstream node $u$, whereas activation under weak pressure fails to do so.

Let $S$ be a small seed set in which only one of $v$’s neighbors is active. Adding $v$ to $S$ causes $v$ to activate immediately but without reinforcement, producing little downstream influence. As a result, the marginal gain $\sigma_{\mathrm{PT}}(S \cup \{v\}) - \sigma_{\mathrm{PT}}(S)$ is small.

Now consider a larger seed set $T \supset S$ that activates multiple neighbors of $v$. In this setting, adding $v$ to $T$ leads to activation under high incoming pressure, triggering reinforcement at activation and substantially amplifying $v$’s outgoing influence. This amplified influence enables $v$ to activate $u$, yielding a larger marginal gain:
\[
\sigma_{\mathrm{PT}}(T \cup \{v\}) - \sigma_{\mathrm{PT}}(T)
\;>\;
\sigma_{\mathrm{PT}}(S \cup \{v\}) - \sigma_{\mathrm{PT}}(S).
\]

Thus, the marginal contribution of node $v$ increases when it is added later rather than earlier. This behavior directly contradicts the diminishing-returns requirement of submodularity and arises from reinforcement-induced complementarities among seed nodes.

The results above establish that influence maximization under the Pressure Threshold model differs fundamentally from its classical counterparts. While reinforcement does not alter the computational hardness of the problem, it eliminates key structural properties that underpin standard approximation theory. In particular, the PT influence function need not be monotone and need not be submodular, as reinforcement introduces context dependence and complementarities that violate diminishing returns.

These findings have important implications for algorithm design. Classical greedy algorithms rely critically on monotonicity and submodularity to guarantee near-optimal performance. When these properties fail, worst-case guarantees no longer apply, and naive application of greedy selection can lead to unpredictable behavior, particularly in marketing and intervention settings where repeated exposure and peer validation shape behavior. At the same time, the structural failures observed here are driven by meaningful behavioral mechanisms rather than pathological constructions, suggesting that purely worst-case analysis may be overly pessimistic.

The loss of classical structure does not preclude effective influence maximization in practice. Instead, it calls for alternative approximation frameworks that explicitly account for reinforcement effects and provide data-dependent performance bounds. In the next section, we develop such an approach by bounding the PT influence function between monotone submodular surrogates, enabling principled greedy heuristics despite the absence of traditional guarantees.

\vspace{.1in} \section{Model Behavior and Approximation Methods}
\label{sec:model_properties}

The preceding section established that influence maximization under the Pressure Threshold (PT) model is computationally hard and lacks the monotonicity and submodularity properties that support classical greedy approximation guarantees. These challenges are particularly relevant in marketing and platform settings, where decision-makers must operate with models that capture reinforcement and peer validation but still admit tractable approximation methods. In this section, we take two complementary steps toward practical decision support under reinforcement-driven diffusion.

First, we provide a compact model demonstration that builds intuition for how activation-triggered amplification alters diffusion dynamics relative to the Linear Threshold (LT) model. This demonstration is not intended to establish general theoretical claims, but rather to illustrate, at a mechanistic level, how reinforcement propagates through a network.

Second, we develop an algorithmic response to the loss of classical structure identified in Section~\ref{sec:decision}. By bounding the PT influence function between monotone submodular surrogates, we derive a sandwich approximation framework that enables principled greedy seed selection despite the absence of traditional worst-case guarantees.

\subsection{Model Demonstration}

To build intuition for the reinforcement dynamics introduced by the Pressure Threshold (PT) model, we present a compact visualization contrasting PT and the classical Linear Threshold (LT) model on a small, controlled instance. The objective is not to claim generality from a toy example, but rather to provide a clear, step-by-step illustration of how additive amplification upon activation alters propagation relative to static-weight diffusion on the same graph.

By holding the network structure, thresholds, and seed set fixed across models, this demonstration isolates the effect of reinforcement and illustrates how local pressure-dependent amplification can accumulate over time. The example serves as a concrete reference point for understanding the structural and algorithmic results developed in later sections. In particular, it highlights how reinforcement transforms diffusion from a purely structural process into one that depends on the history and intensity of activation.

We illustrate the PT dynamics using a small-scale example with amplification parameter $\alpha = 0.25$ on a $9$-node network. Figures~\ref{fig:ptlt-early} and~\ref{fig:ptlt-late} visualize the evolution of PT diffusion alongside LT on the same graph and with an identical seed set. A moderately large value of $\alpha$ is chosen to highlight the qualitative distinction between reinforcement-driven and static diffusion while keeping the example easy to trace visually.

Both models begin with two identical seed nodes at time $t = 0$. Activated nodes are shown in red, newly activated nodes are highlighted with a black outline, and unactivated nodes are shown in blue. In the PT model, edges whose weights are amplified during the influence adjustment phase are highlighted in red. Node labels display the node index (underlined), together with the corresponding activation threshold $\theta$.

In the earliest timesteps, PT and LT behave similarly, since all influence weights are initially identical and no amplification has yet occurred. Differences begin to emerge once nodes activate through propagation in PT and receive reinforcement (Figures~\ref{fig:ptlt-early}). As amplification accumulates across rounds, PT enables additional nodes to cross their activation thresholds, whereas the LT cascade stalls earlier (Figure~\ref{fig:ptlt-late}). In this example, PT ultimately achieves full network coverage, while LT terminates after activating only five nodes.


\begin{figure*}[htbp]
\centering
\begin{subfigure}[t]{0.48\textwidth}
    \centering
    \begin{subfigure}[t]{\linewidth}
        \includegraphics[width=\linewidth]{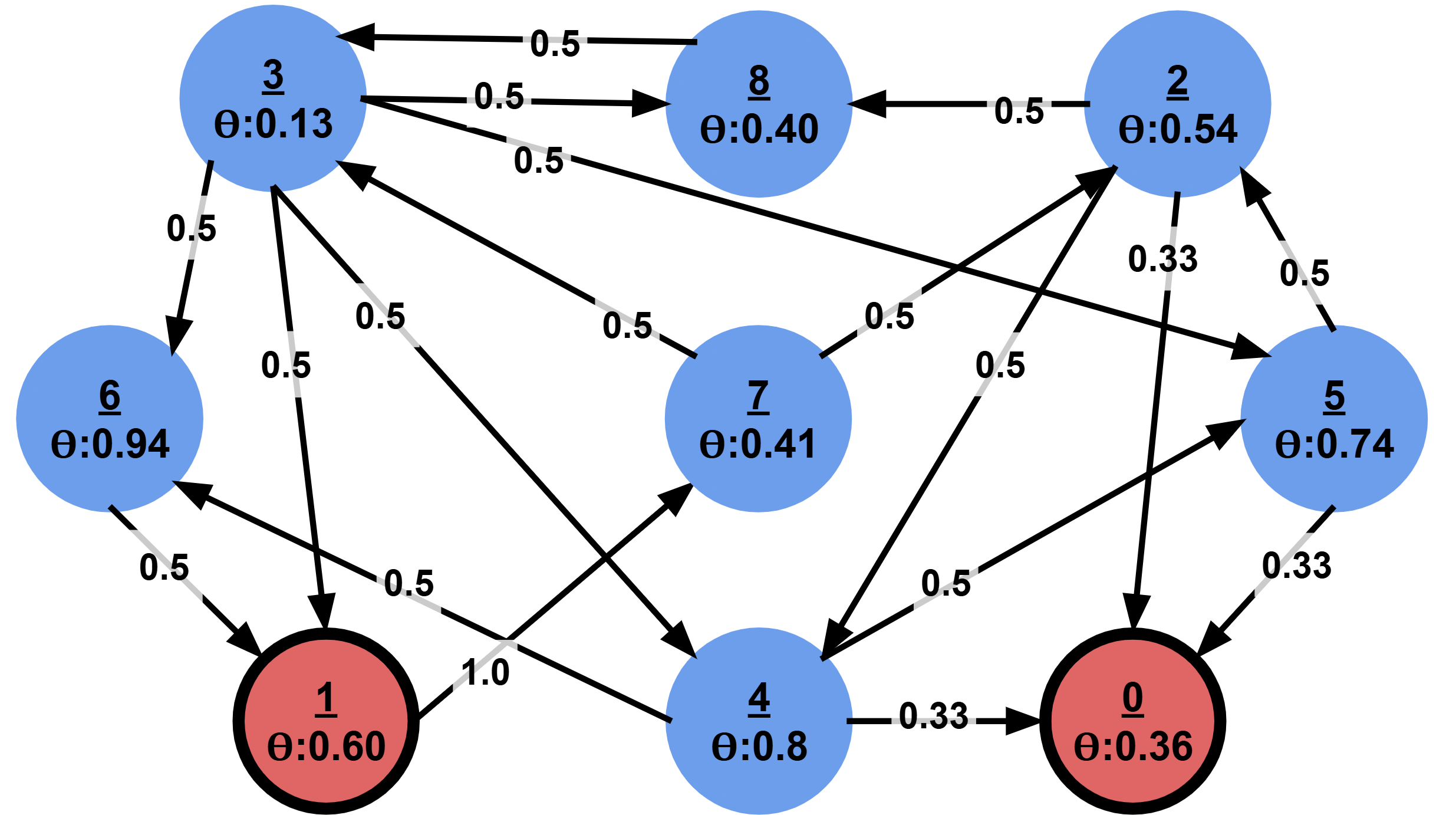}
        \vspace{.5pt}
        \caption*{PT Step 0 (No. of activated nodes = 2)}
        \label{fig:pt-step0}
    \end{subfigure}\vspace{35pt}
    \begin{subfigure}[t]{\linewidth}
        \includegraphics[width=\linewidth]{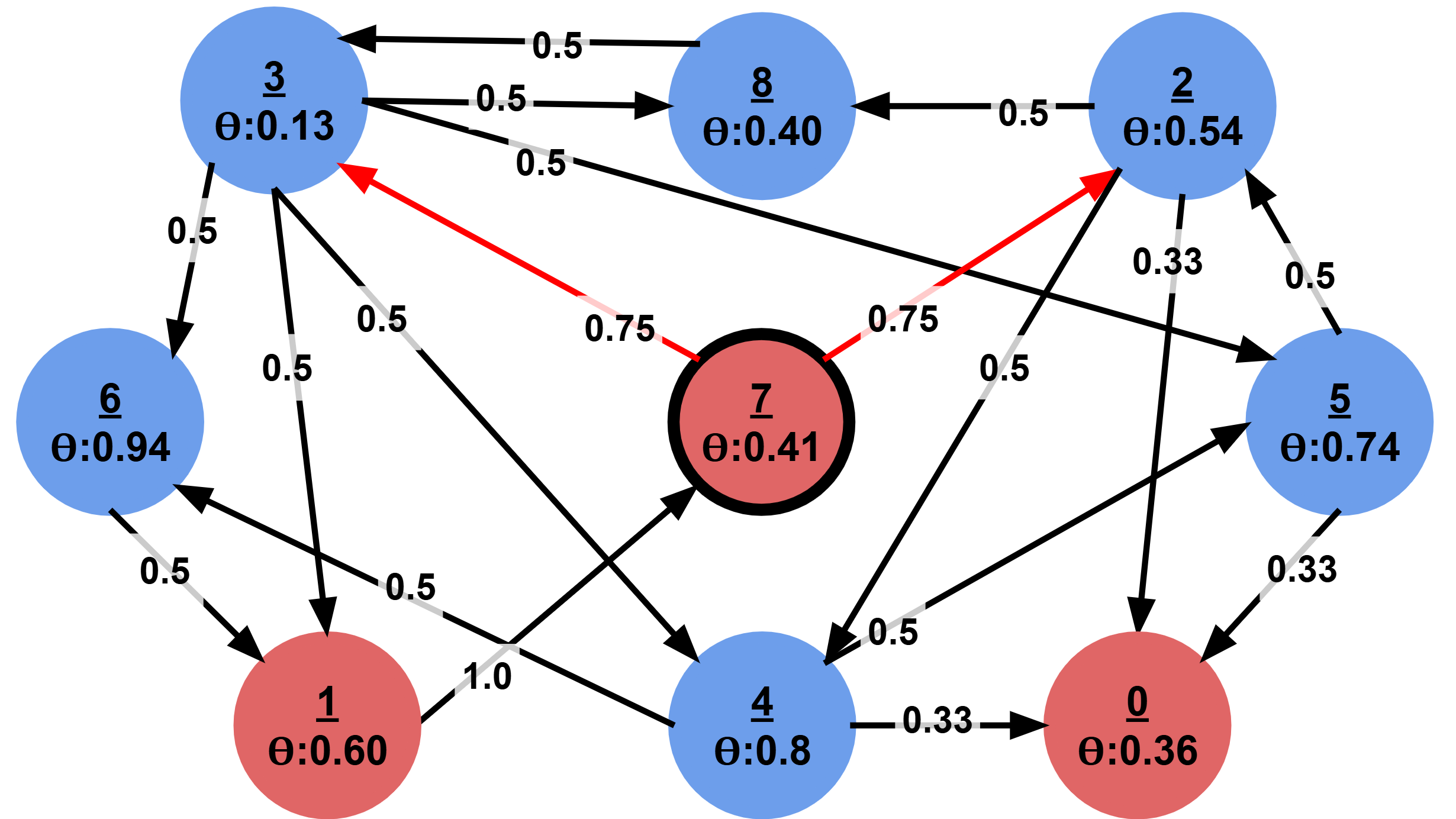}
        \vspace{.5pt}
        \caption*{PT Step 1 (No. of activated nodes = 3)}
        \label{fig:pt-step1}
    \end{subfigure}\vspace{35pt}
    \begin{subfigure}[t]{\linewidth}
        \includegraphics[width=\linewidth]{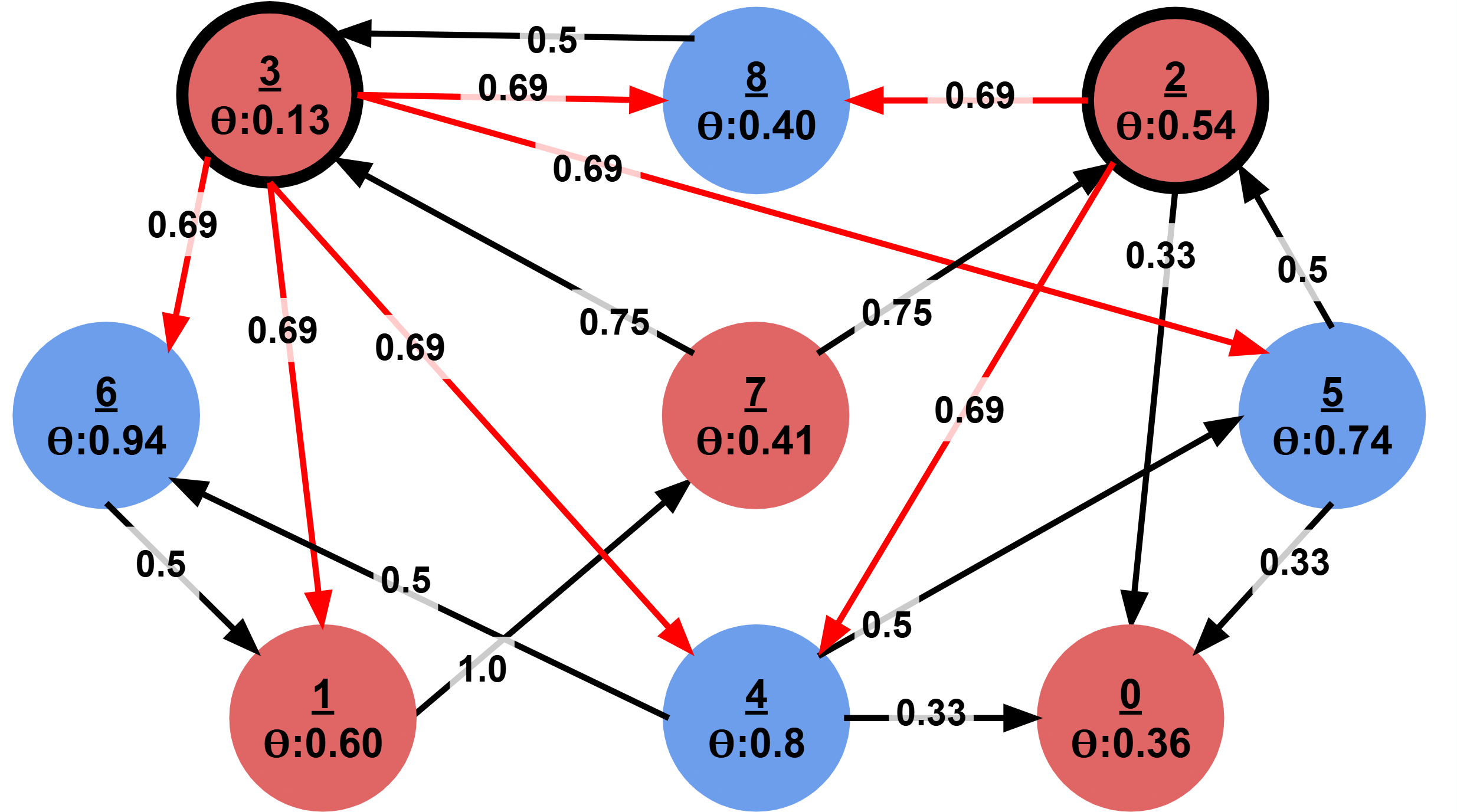}
        \vspace{.5pt}
        \caption*{PT Step 2 (No. of activated nodes = 5)}
        \label{fig:pt-step2}
    \end{subfigure}\vspace{10pt}
\end{subfigure}\hfill
\begin{subfigure}[t]{0.48\textwidth}
    \centering
    \begin{subfigure}[t]{\linewidth}
        \includegraphics[width=\linewidth]{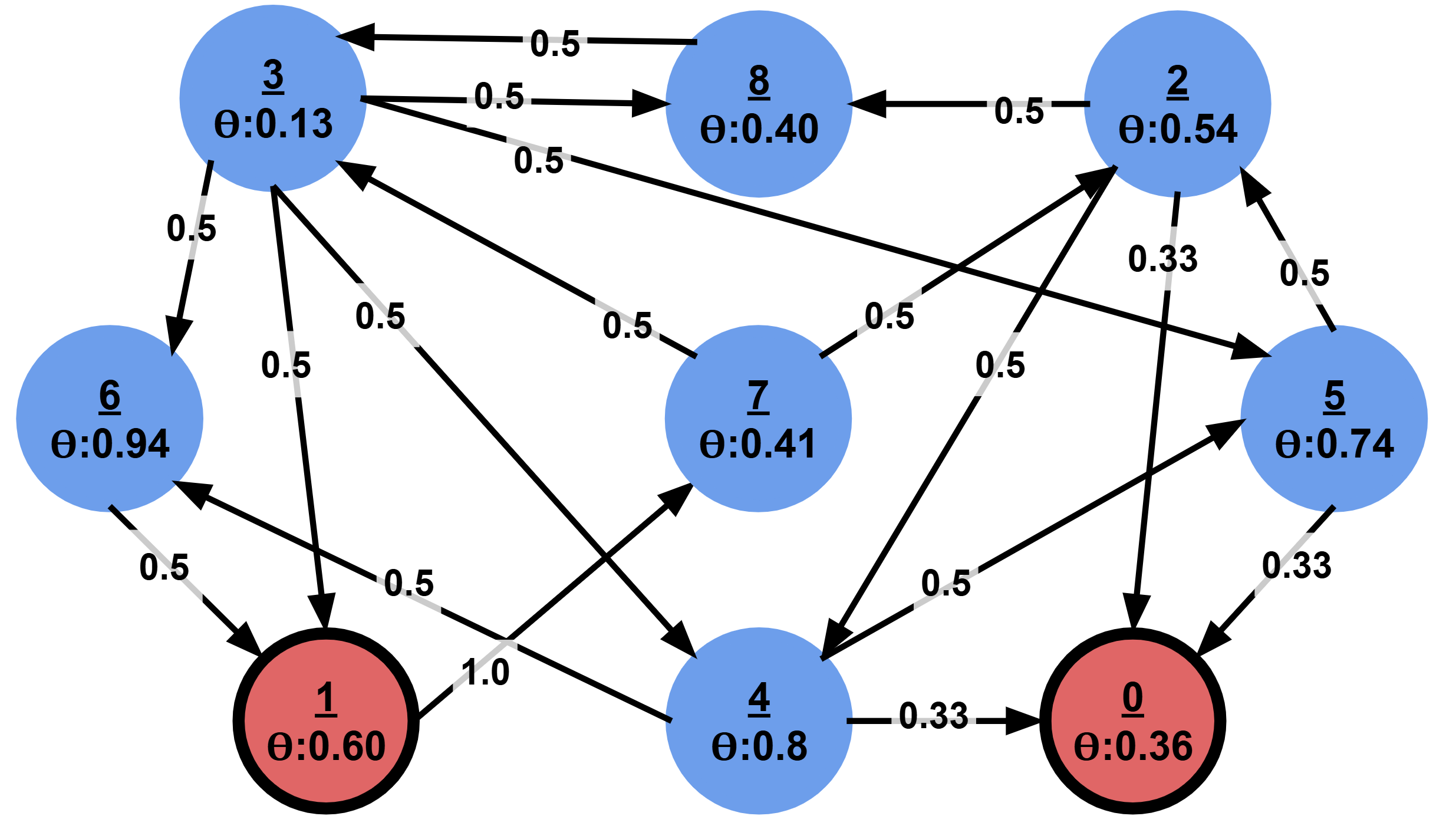}
        \vspace{.5pt}
        \caption*{LT Step 0 (No. of activated nodes = 2)}
        \label{fig:lt-step0}
    \end{subfigure}\vspace{35pt}
    \begin{subfigure}[t]{\linewidth}
        \includegraphics[width=\linewidth]{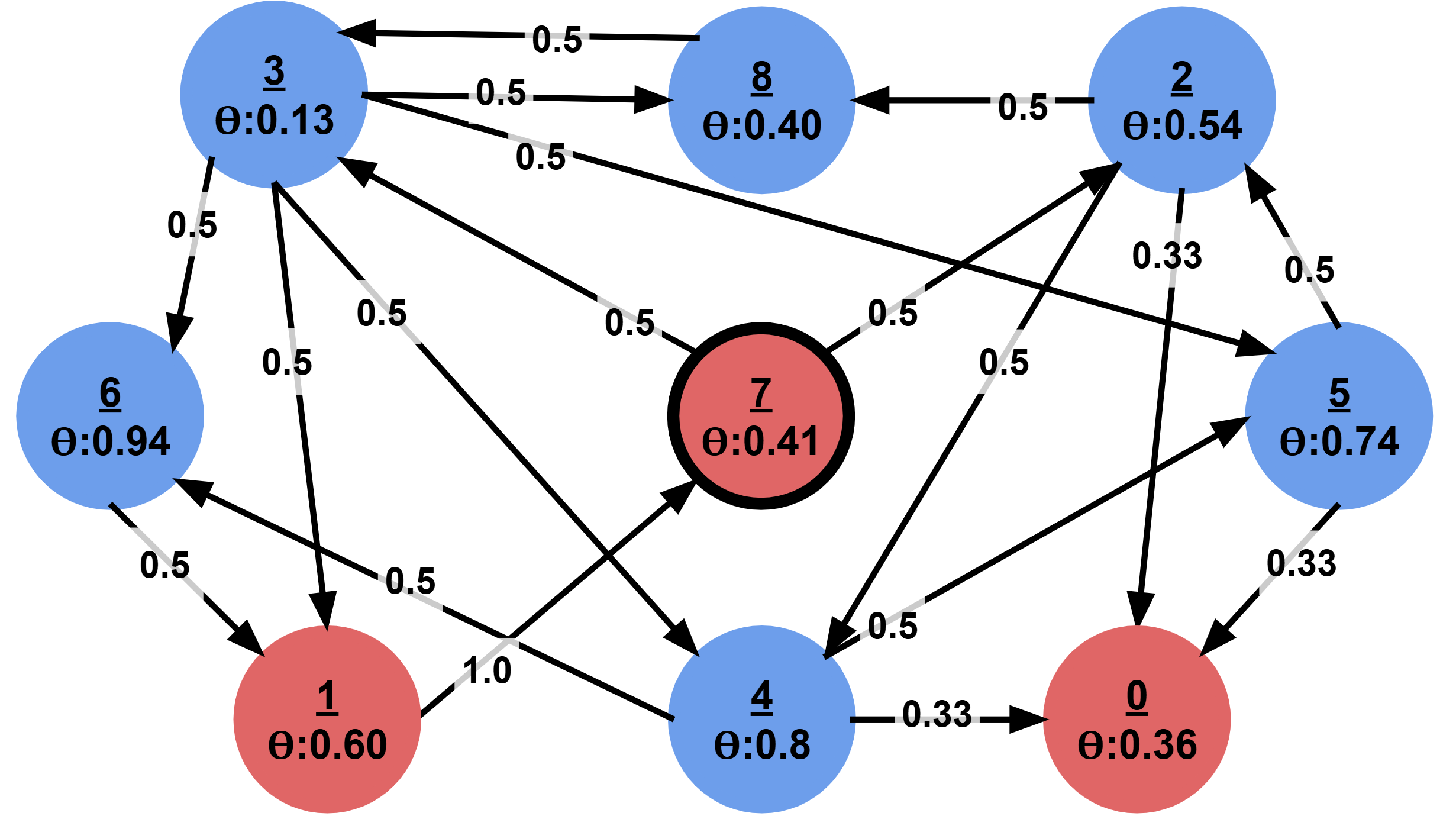}
        \vspace{.5pt}
        \caption*{LT Step 1 (No. of activated nodes = 3)}
        \label{fig:lt-step1}
    \end{subfigure}\vspace{35pt}
    \begin{subfigure}[t]{\linewidth}
        \includegraphics[width=\linewidth]{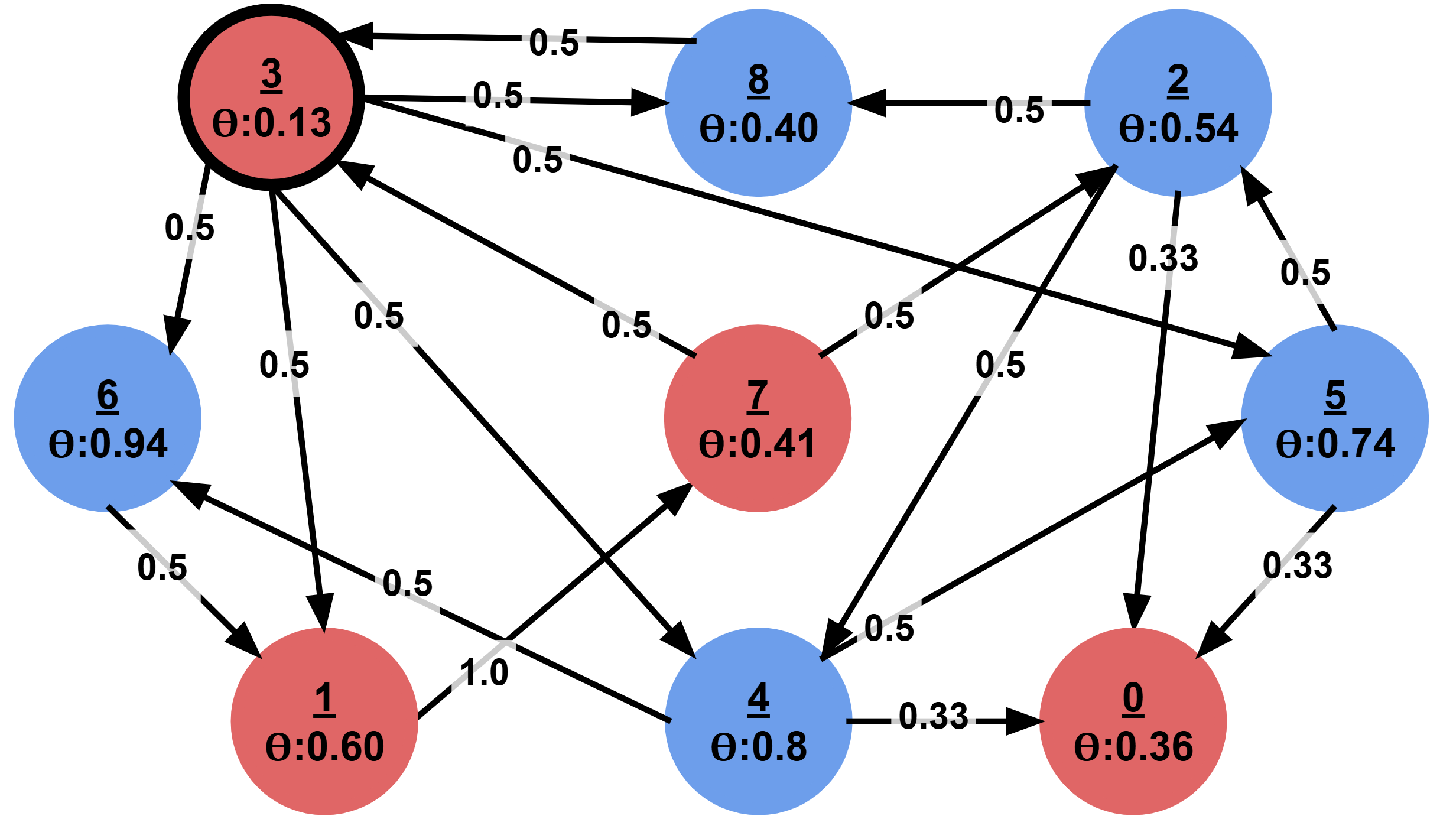}
        \vspace{.5pt}
        \caption*{LT Step 2 (No. of activated nodes = 4)}
        \label{fig:lt-step2}
    \end{subfigure}\vspace{10pt}
\end{subfigure}
\caption{Early diffusion stages under PT and LT on the same graph using an identical seed set. Visualizations show the initial evolution of the diffusion process, during which the reinforcement effects are present but have not yet accumulated strongly. Activated nodes are shown in red, newly activated nodes are highlighted in red with a black outline, and unactivated nodes are shown in blue. Newly amplified edges resulting from PT’s influence adjustment mechanism are highlighted in red. Node labels display the node index (underlined) together with the corresponding activation threshold $\theta$.}
\label{fig:ptlt-early}
\end{figure*}


\begin{figure*}[htbp]
\centering
\begin{subfigure}[t]{0.48\textwidth}
    \centering
    \begin{subfigure}[t]{\linewidth}
        \includegraphics[width=\linewidth]{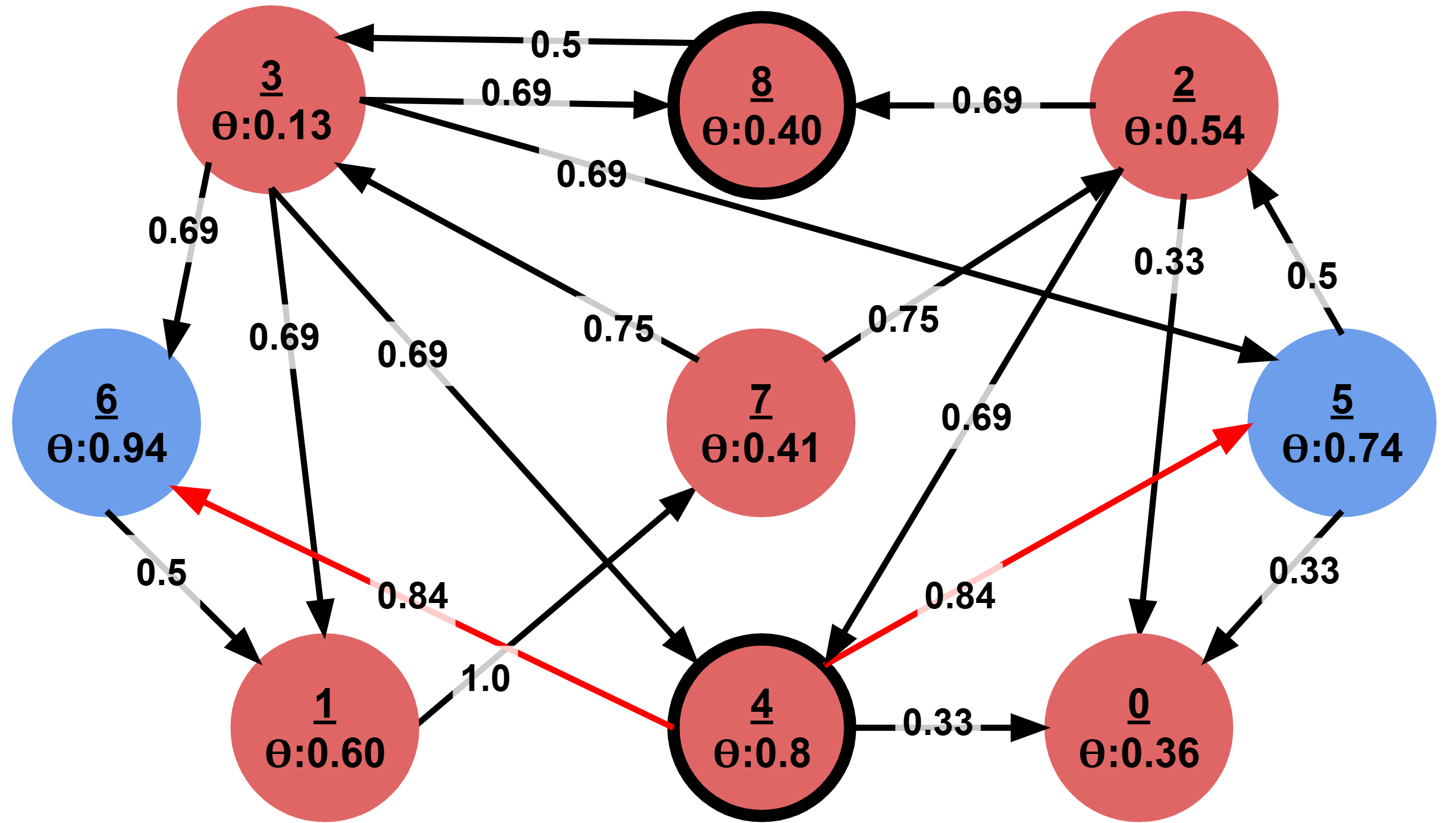}
        \vspace{.5pt}
        \caption*{PT Step 3 (No. of activated nodes = 7)}
        \label{fig:pt-step3}
    \end{subfigure}\vspace{35pt}
    \begin{subfigure}[t]{\linewidth}
        \includegraphics[width=\linewidth]{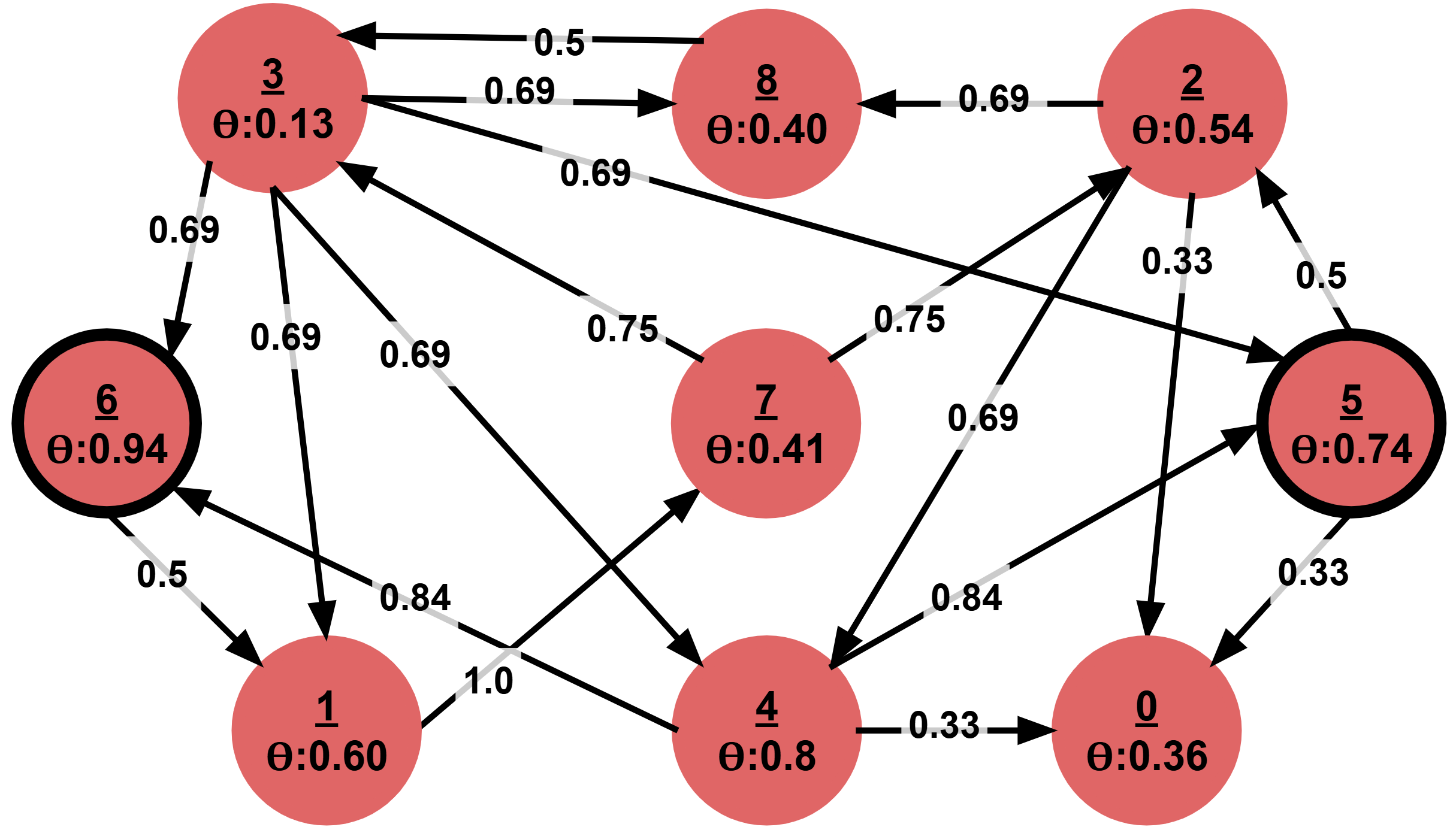}
        \vspace{.5pt}
        \caption*{PT Step 4 (No. of activated nodes = 9)}
        \label{fig:pt-step4}
    \end{subfigure}\vspace{35pt}
    \begin{subfigure}[t]{\linewidth}
        \includegraphics[width=\linewidth]{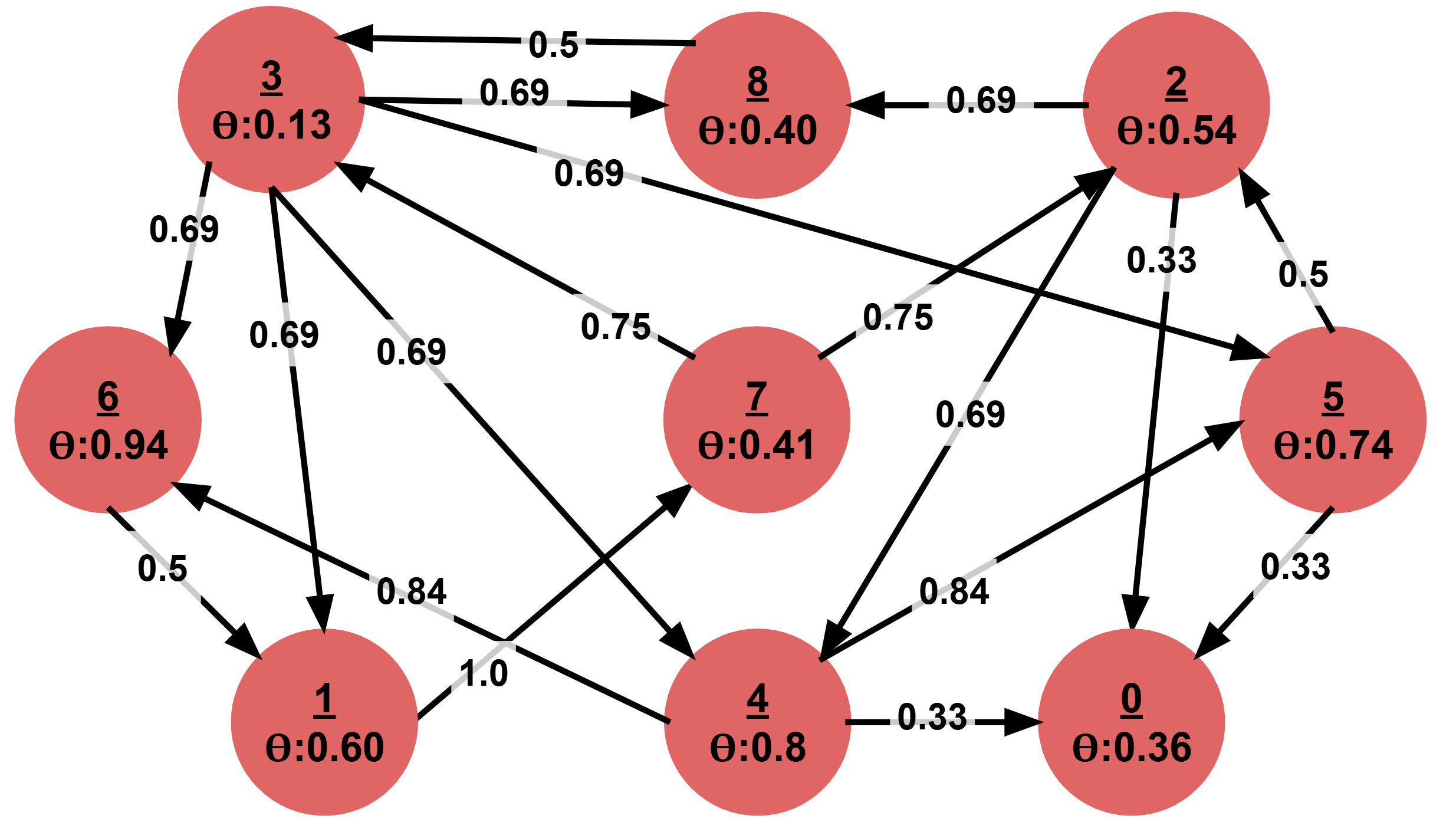}
        \vspace{.5pt}
        \caption*{PT Final Step (No. of activated nodes = 9)}
        \label{fig:pt-step5}
    \end{subfigure}\vspace{10pt}
\end{subfigure}\hfill
\begin{subfigure}[t]{0.48\textwidth}
    \centering
    \begin{subfigure}[t]{\linewidth}
        \includegraphics[width=\linewidth, trim=.05pt .05pt .05pt .05pt, clip]{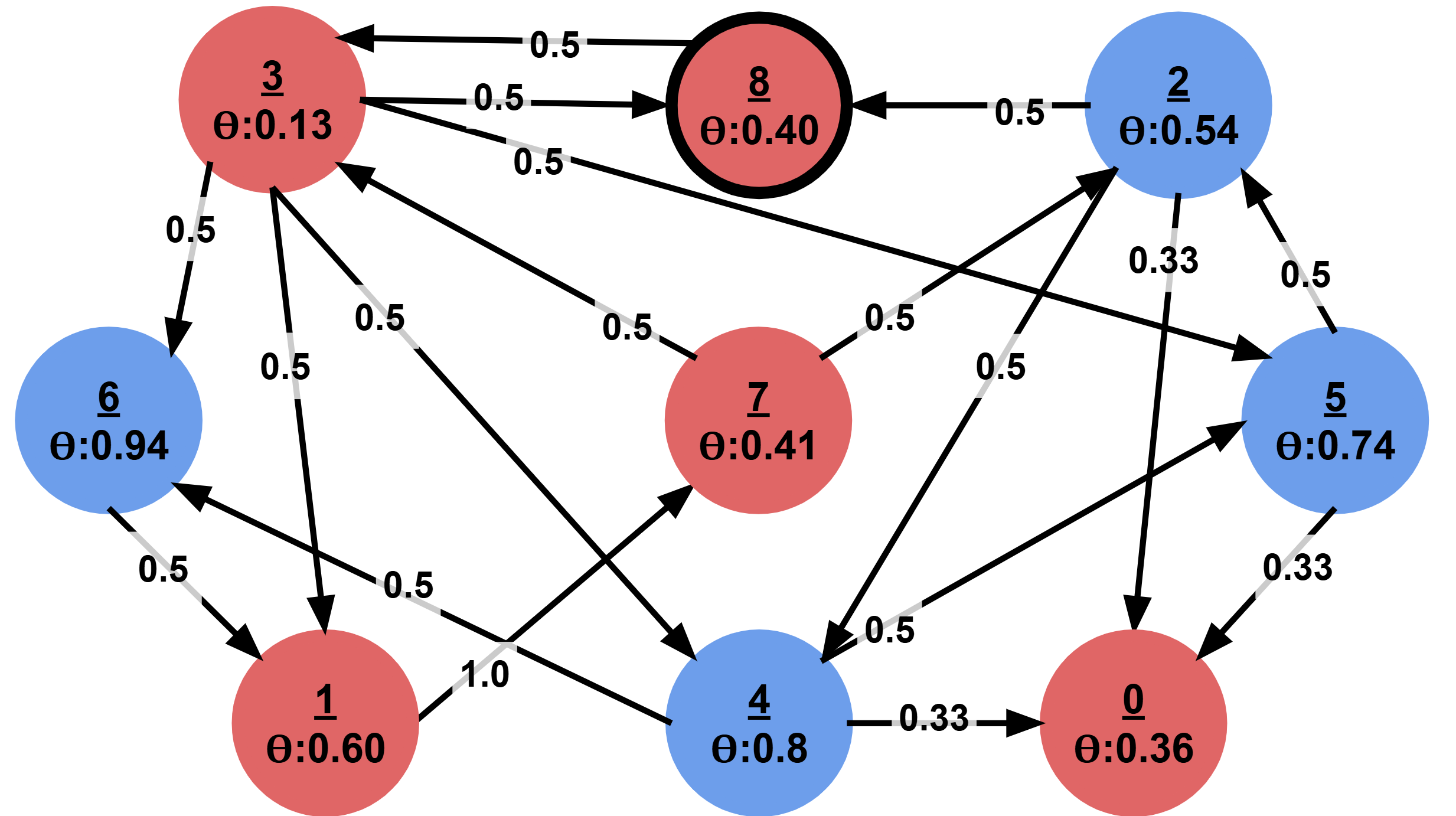}
        \vspace{.5pt}
        \caption*{LT Step 3 (No. of activated nodes = 5)}
        \label{fig:lt-step3}
    \end{subfigure}\vspace{35pt}
    \begin{subfigure}[t]{\linewidth}
        \includegraphics[width=\linewidth, trim=0pt 0pt 0pt 0pt, clip]{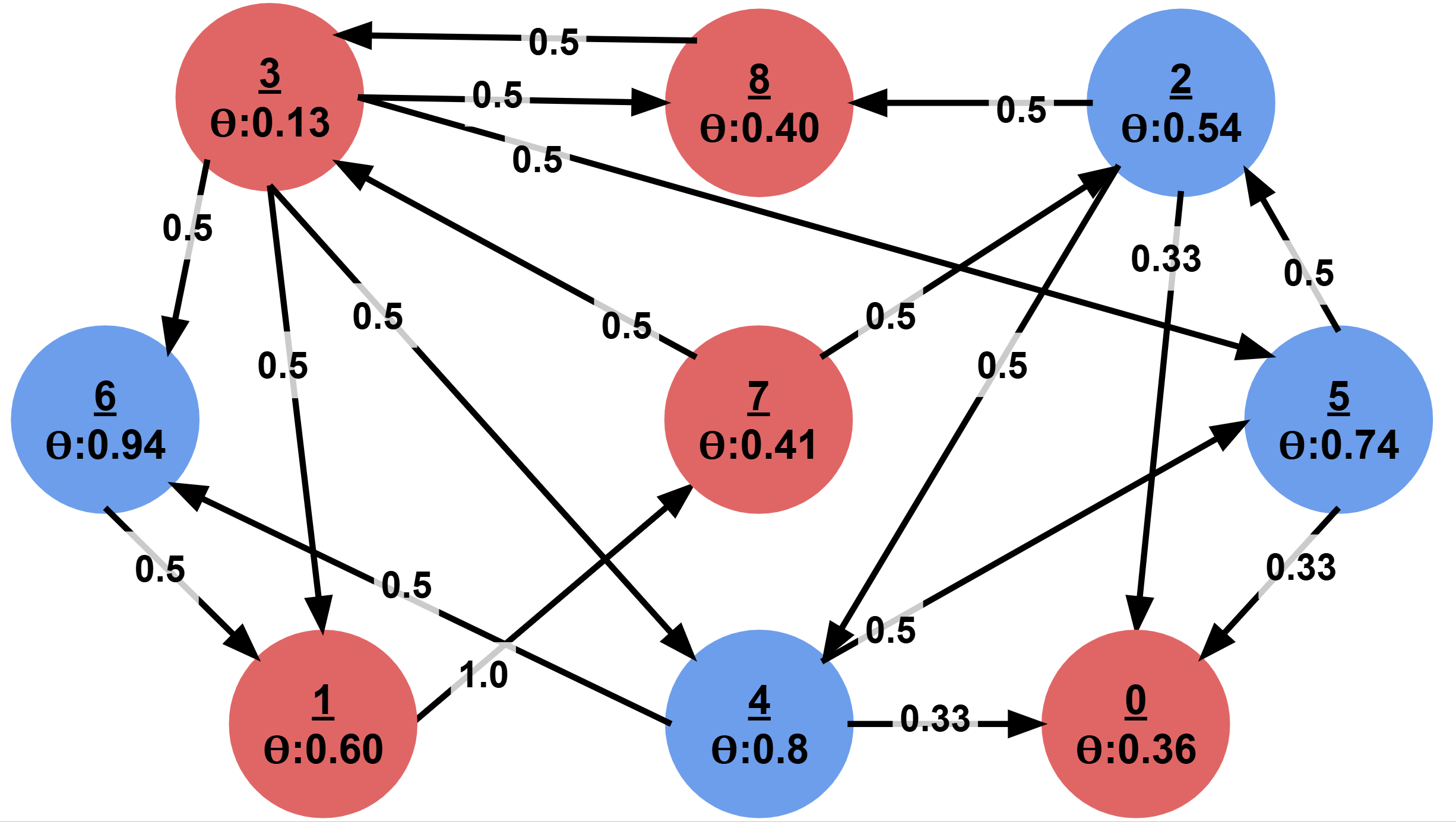}
        \vspace{.5pt}
        \caption*{LT Final Step (No. of activated nodes = 5)}
        \label{fig:lt-final}
    \end{subfigure}\vspace{35pt}
    \begin{subfigure}[t]{\linewidth}
    \end{subfigure}\vspace{10pt}
\end{subfigure}
\caption{Later diffusion stages under PT and LT on the same graph using an identical seed set. As the diffusion evolves, reinforcement effects intensify. PT continues amplifying outgoing influence along activated edges and ultimately attains full network coverage, whereas LT halts after activating fewer nodes. Activated nodes are shown in red, newly activated nodes are highlighted in red with a black outline, and unactivated nodes are shown in blue. Newly amplified edges due to PT’s reinforcement mechanism are highlighted in red. Node labels display the node index (underlined) together with the corresponding activation threshold $\theta$.}
\label{fig:ptlt-late}
\end{figure*}

This simple example illustrates how activation-triggered reinforcement can qualitatively alter diffusion outcomes, even when the network structure, thresholds, and initial seeds are fixed. In particular, it highlights how local amplification effects can accumulate over time, enabling cascades that would otherwise stall under static-weight diffusion. While this example is not intended to characterize typical or worst-case behavior, it provides a concrete reference point for understanding how reinforcement manifests in the PT model. This shift from static to reinforcement-driven diffusion is precisely what underlies the structural breakdown of monotonicity and submodularity established earlier.


\subsection{Consequences of Reinforcement for Optimization}

As established in Section~\ref{sec:decision}, influence maximization under the PT model remains NP-hard and lacks the monotonicity and submodularity properties that support classical greedy guarantees. These structural violations arise from the same reinforcement mechanism illustrated in the preceding model demonstration: activation-triggered amplification introduces context dependence and complementarities that are incompatible with diminishing returns. From a decision perspective, this implies that the value of targeting a given individual depends critically on the broader set of targeted individuals, rather than being separable across nodes.

Rather than abandoning greedy approaches altogether, we develop an approximation framework that explicitly accounts for reinforcement effects. Our strategy is to bound the PT influence function between two tractable, monotone submodular functions. These bounds provide a principled way to evaluate greedy solutions under PT, even when the true objective lacks favorable structure.

To address these challenges, we construct tractable surrogate models that preserve key aspects of PT dynamics while restoring structural properties that enable efficient optimization.

\subsubsection{Static Maximum-Pressure Model Definition}

We begin by defining a static upper-bounding diffusion model, denoted $M^+$. The $M^+$ model is defined on the same graph structure as the PT and LT models and uses the same initial parameters, but replaces PT’s dynamic reinforcement mechanism with a static, worst-case amplification applied uniformly at the start of the diffusion process.

Let $G=(V,E)$ be a directed graph, where $V$ is the set of nodes and $E \subseteq V \times V$ is the set of directed edges. Each edge $(u,v)\in E$ is assigned a baseline weight $w_{uv} \in [0,1]$ such that $\sum_{u \in N(v)} w_{uv} \le 1$ for each node $v$, where $N(v)$ denotes the in-neighbors of $v$. As in our earlier analyses, we initialize each weight according to the weighted-cascade convention, $w_{uv}=1/\text{in-degree}(v)$. Each node $v \in V$ is assigned an activation threshold $\theta_v \in (0,1]$.

The key observation underlying $M^+$ is that under PT diffusion, the total incoming pressure experienced by a node at activation is bounded above by $1$ under the weighted-cascade convention. This construction provides a worst-case envelope of reinforcement effects, ensuring that $M^+$ captures the strongest possible amplification consistent with the PT mechanism. Consequently, the maximum possible amplification applied to any outgoing edge $(u,v)$ during a PT cascade is $\alpha$. The $M^+$ model, therefore, replaces the dynamic, activation-triggered updates of PT with a static pre-amplified weight defined as
\[
w_{uv}^+ = \min\{1,\; w_{uv} + \alpha\}.
\]

Diffusion under the $M^+$ model proceeds exactly as in the Linear Threshold framework, with the exception that static pre-amplified weights $w_{uv}^+$ are used throughout the process. At each time step, an inactive node $v$ becomes active if the cumulative influence from its active in-neighbors meets or exceeds its threshold:
\[
\sum_{u \in P(v)} w_{uv}^+ \;\ge\; \theta_v.
\]
As in LT and PT, the process is progressive and terminates when no further nodes activate.

\begin{lemma}\label{lem:mplus-dominates}
The $M^+$ model dominates the PT model in expected influence: for all seed sets $S \subseteq V$,
\[
\sigma_{\mathrm{PT}}(S) \;\le\; \sigma_{M^+}(S).
\]
\end{lemma} 

\proof \\
Fix an arbitrary seed set $S$ and a realization of thresholds $\{\theta_v\}_{v \in V}$. Run both the PT and $M^+$ diffusion processes on the same graph with these thresholds.

Let $A_t^{\mathrm{PT}}$ and $A_t^{M^+}$ denote the active sets at time $t$ under PT and $M^+$, respectively. We prove by induction on $t$ that $A_t^{\mathrm{PT}} \subseteq A_t^{M^+}$.

In the base case, $A_0^{\mathrm{PT}} = A_0^{M^+} = S$. Assume $A_t^{\mathrm{PT}} \subseteq A_t^{M^+}$. Consider any node $v \notin A_t^{\mathrm{PT}}$. Under PT, $v$ activates at time $t+1$ if the sum of dynamic incoming weights from $A_t^{\mathrm{PT}}$ meets or exceeds $\theta_v$. Under $M^+$, $v$ has at least as many active parents, and each corresponding incoming weight satisfies $w_{uv}^+ \ge w'_{uv}$. Therefore, any node that activates under PT must also activate under $M^+$.

Thus $A_{t+1}^{\mathrm{PT}} \subseteq A_{t+1}^{M^+}$, and by induction the inclusion holds at termination. Since this holds pointwise for every threshold realization, it also holds in expectation, implying $\sigma_{\mathrm{PT}}(S) \le \sigma_{M^+}(S)$. This establishes that $M^+$ provides a valid upper bound on PT diffusion for all seed sets.
\endproof

\begin{remark}
Because the $M^+$ model uses static edge weights and follows the Linear Threshold activation rule, its influence function $\sigma_{M^+}(S)$ is monotone and submodular \citep{kempe2003maximizing}. As such, $\sigma_{M^+}$ admits a $(1-1/e)$ approximation guarantee via greedy maximization.
\end{remark}

\subsubsection{Sandwich Approximation Formulation}

We now formalize the sandwich approximation strategy. Let $\sigma_{\mathrm{LT}}$ denote the influence function under the classical Linear Threshold model, $\sigma_{\mathrm{PT}}$ under the Pressure Threshold model, and $\sigma_{M^+}$ under the static maximum-pressure model. For any seed set $S \subseteq V$, these influence functions satisfy
\[
\sigma_{\mathrm{LT}}(S) \;\le\; \sigma_{\mathrm{PT}}(S) \;\le\; \sigma_{M^+}(S).
\]

Thus, PT diffusion is sandwiched between two monotone submodular processes that serve as tractable lower and upper bounds.

Since both $\sigma_{\mathrm{LT}}$ and $\sigma_{M^+}$ are monotone submodular, we may run the standard greedy algorithm on each to obtain seed sets $S_{\mathrm{LT}}$ and $S_{M^+}$ with $(1-1/e)$ approximation guarantees. We additionally consider the seed set $S_{\mathrm{PT}}$ obtained by running greedy directly on $\sigma_{\mathrm{PT}}$.

The sandwich strategy selects the final seed set as
\[
S_{\mathrm{sand}}
\;=\;
\arg\max_{S \in \{S_{\mathrm{LT}},\, S_{\mathrm{PT}},\, S_{M^+}\}} \;
\sigma_{\mathrm{PT}}(S).
\]

\subsubsection{Approximation Guarantee}

Let $S_{\mathrm{PT}}^\star$ denote an optimal PT seed set. By standard sandwich arguments \citep{Lu2015}, the selected seed set satisfies
\[
\sigma_{\mathrm{PT}}(S_{\mathrm{sand}})
\;\ge\;
(1-1/e)
\cdot
\max\!\left\{
\frac{\sigma_{\mathrm{PT}}(S_{M^+})}{\sigma_{M^+}(S_{M^+})},\;
\frac{\sigma_{\mathrm{LT}}(S_{\mathrm{PT}}^\star)}{\sigma_{\mathrm{PT}}(S_{\mathrm{PT}}^\star)}
\right\}
\cdot
\sigma_{\mathrm{PT}}(S_{\mathrm{PT}}^\star).
\]

While the second term depends on the unknown optimum $S_{\mathrm{PT}}^\star$, the first term is fully computable. This yields a practical, data-dependent approximation guarantee of
\[
(1-1/e)\cdot
\frac{\sigma_{\mathrm{PT}}(S_{M^+})}{\sigma_{M^+}(S_{M^+})},
\]
which can be evaluated efficiently using Monte Carlo estimation of influence.

This framework demonstrates that even in the absence of classical structural guarantees, reinforcement-aware diffusion models can still support principled and effective decision-making.

\vspace{.1in} \section{Computational Experiments}
\label{sec:experiments}

We now evaluate the implications of pressure-based reinforcement for diffusion dynamics and network-based interventions. Our experiments are designed to address three central questions: 
(i) how reinforcement alters seed selection in influence maximization, 
(ii) how it impacts realized diffusion outcomes across different network structures, and 
(iii) how sensitive these effects are to the strength of social reinforcement.

Consistent with the decision-oriented perspective of this paper, our goal is not only to compare the Pressure Threshold (PT) model with the Linear Threshold (LT) model, but also to understand how reinforcement reshapes the effectiveness of intervention strategies in practice. We therefore emphasize both behavioral differences in diffusion and their consequences for decision quality. In particular, our experiments are designed to assess when accounting for reinforcement leads to meaningfully different intervention decisions compared to classical diffusion assumptions.

Across all experiments, we follow a common methodology. We evaluate diffusion on a combination of real-world and synthetic networks, use Monte Carlo simulation to estimate influence, and compare PT against LT under matched conditions to isolate the effect of reinforcement. Unless otherwise stated, all results reflect progressive diffusion processes and average outcomes across repeated simulations. All comparisons are conducted under identical initial conditions to ensure that differences in outcomes are driven solely by the presence of reinforcement.

\subsection{Network Data}

We evaluate the Pressure Threshold (PT) model against the Linear Threshold (LT) model on four networks spanning distinct topologies and application domains. Three of these are real-world networks drawn from the Stanford Large Network Dataset Collection (SNAP) \cite{snapnets}: 
(i) Facebook social circles \cite{leskovec2012learning}, 
(ii) Wikipedia voting \cite{leskovec2010predicting,leskovec2010signed}, and 
(iii) Bitcoin OTC trust \cite{kumar2016edge,kumar2018rev2}. 
To complement these datasets, we also include a synthetic Erd\H{o}s--R\'enyi random network \(G(n,p)\) \cite{erdds1959random}.

These networks were selected to capture variation in structural properties that are expected to interact with reinforcement, including density, clustering, and directionality of influence. These properties directly influence the extent to which nodes are exposed to repeated influence, a key driver of amplification under the PT model. In particular, denser networks provide more opportunities for repeated exposure and therefore greater potential for amplification under PT. 

Table~\ref{tab:basic_info} summarizes key statistics for each network, including the number of nodes, number of edges, and edge-to-node ratio. This variation allows us to examine how the effects of reinforcement depend on underlying network structure.

\begin{table}[h] 
    \centering
    \begin{tabular}{l c c c}
        \hline \hline
        Network & Nodes & Edges & Edge/Node Ratio\\
        \hline
        Facebook      & 4{,}039 & 88{,}234  & 21.846 \\
        Bitcoin       & 5{,}881 & 35{,}592  & 6.052  \\
        Wikipedia     & 7{,}115 & 103{,}689 & 14.573 \\
        Erd\H{o}s--R\'enyi  & 5{,}000 & 62{,}597  & 12.519 \\
        \hline
        \hline
    \end{tabular}
    \caption{Basic statistics of the networks used in experiments.}
    \label{tab:basic_info}
\end{table}

\paragraph{Graph construction and directionality.}
The four networks differ in how edges encode influence relationships, and we standardize them to a unified directed framework consistent with threshold-based diffusion. 

The Facebook social circles dataset is originally undirected, with an edge indicating mutual friendship. Following standard practice in diffusion modeling \cite{kempe2003maximizing}, we convert each undirected edge $\{u,v\}$ into two directed edges $(u,v)$ and $(v,u)$, preserving symmetric opportunities for influence. 

The Bitcoin OTC trust and Wikipedia voting networks are inherently directed. In the Bitcoin network, an edge $(u,v)$ represents a trust rating assigned by user $u$ to user $v$, while in the Wikipedia network, an edge $(u,v)$ indicates that user $u$ voted for user $v$ in an administrative election. These interpretations naturally align with directed influence flow.

For the Erd\H{o}s--R\'enyi network, we generate a random graph $G(n,p)$ with $n=5{,}000$ and $p=0.005$, and similarly convert each undirected edge into two directed edges. This provides a controlled baseline without community structure, allowing us to isolate the role of network topology in shaping reinforcement effects.

\paragraph{Edge weights and normalization.}
We adopt the weighted-cascade convention \cite{kempe2003maximizing}, assigning each directed edge $(u,v)$ the weight
\[
w_{uv} \;=\; \frac{1}{\text{in-degree}(v)}.
\]
This normalization ensures that the total incoming influence to any node is bounded by one, i.e., $\sum_{u \in N(v)} w_{uv} \le 1$, which is consistent with the interpretation of thresholds as adoption constraints. This normalization is particularly important under PT, as it ensures that reinforcement operates within a bounded and interpretable range of influence.

Under the Pressure Threshold (PT) model, these baseline weights evolve endogenously through the reinforcement mechanism described in Section~\ref{sec:model}. In particular, when a node activates, its outgoing influence weights may increase as a function of the incoming pressure it experienced. All updates respect the cap $w_{uv} \le 1$, ensuring comparability with LT and interpretability of influence magnitudes.

\paragraph{Thresholds and stochasticity.}
In all experiments, node thresholds $\theta_v$ are drawn independently from the uniform distribution on $(0,1]$, following the standard assumption in threshold-based diffusion models. This choice provides a neutral baseline that does not favor either early or late adoption and allows comparisons across models to focus on differences in influence dynamics.

To ensure fair comparisons between PT and LT, we use matched randomness: for any given seed set and trial, the same realization of thresholds is used under both models. This pairing isolates the impact of reinforcement by controlling for stochastic variation in activation thresholds. As a result, differences between PT and LT can be interpreted causally as arising from the reinforcement mechanism rather than stochastic variation.

All reported influence values are computed as averages over independent Monte Carlo simulations. Unless otherwise specified, each estimate is based on 1{,}000 trials, which provides stable estimates of expected spread and allows consistent comparison across networks and parameter settings.

\paragraph{Interpretation for decision-making.}
This experimental design isolates the effect of reinforcement while holding the underlying network structure and baseline influence assumptions fixed. As a result, differences in observed diffusion outcomes can be attributed directly to the presence or absence of pressure-based amplification.

For decision makers, this setup provides a direct analogue of intervention strategies designed under static influence assumptions (LT) versus those that account for conviction and reinforcement (PT). The experiments, therefore, provide insight not only into diffusion behavior but also into how modeling choices affect the evaluation and selection of network interventions.

\subsection{Experimental Design}

To understand how reinforcement affects both diffusion outcomes and intervention strategies, we organize our empirical analysis around three complementary questions:

\begin{enumerate}
    \item \textbf{How does reinforcement alter seed selection?} 
    Do strategies optimized under PT select different nodes compared to LT, and if so, how do these differences emerge?

    \item \textbf{How does reinforcement impact diffusion outcomes?}
    For a given budget, how much additional influence does PT generate relative to LT across networks with different structural properties?

    \item \textbf{How sensitive are outcomes to the strength of reinforcement?}
    How does variation in the amplification parameter $\alpha$ affect the magnitude and stability of diffusion outcomes?
\end{enumerate}

Together, these questions allow us to disentangle how reinforcement affects both the selection of intervention targets and the realized impact of those interventions.

We design a set of experiments to isolate each of these questions while maintaining consistent assumptions across models.

\paragraph{Seed selection under reinforcement.}
To examine how reinforcement reshapes intervention strategies, we compare seed sets obtained under the PT and LT models on the Facebook network. Specifically, we solve the influence maximization problem using the CELF algorithm \cite{leskovec2007cost} with a seed budget of $k=20$. 

Influence is estimated using $1{,}000$ Monte Carlo simulations per marginal gain evaluation. For PT, we set $\alpha = 0.1$, which empirically provides a regime in which reinforcement effects are visible but do not induce immediate saturation. To ensure comparability, we use the same threshold realizations for both PT and LT during seed selection. 

Rather than focusing only on final influence values, we track the full sequence of node selections made by CELF. This allows us to examine how reinforcement affects marginal gains over time and where divergences between models emerge. In particular, it allows us to identify when reinforcement begins to alter marginal contributions in ways that deviate from classical diminishing-returns intuition.

\paragraph{Influence as a function of budget.}
To quantify the impact of reinforcement on realized diffusion outcomes, we evaluate influence as a function of the seed budget across all four networks. For each budget $k \in \{1,2,\dots,60\}$, we run CELF under three settings: LT, PT with $\alpha = 0.001$, and PT with $\alpha = 0.005$.

Each influence evaluation during the optimization process is computed using $1{,}000$ Monte Carlo simulations. The selected values of $\alpha$ are chosen to illustrate meaningful deviations from LT without inducing trivial full-network cascades, particularly in denser graphs.

This experiment allows us to characterize how reinforcement shifts the return on additional seeds and how these effects depend on network topology.

\begin{remark}[Robustness to threshold distributions]
We evaluated alternative threshold distributions, including Normal$(0.5,0.1)$ and Beta distributions with parameters $(2,5)$ and $(5,2)$. Across all cases, the qualitative differences between PT and LT remained consistent. For brevity, we report results under the standard uniform prior.
\end{remark}

\paragraph{Sensitivity to the reinforcement parameter.}
To isolate the effect of reinforcement strength, we fix a seed set size of $k=10$ on the Facebook network and vary the amplification parameter $\alpha$ over the interval $[10^{-4}, 10^{-1}]$ using a fine grid.

For each value of $\alpha$, we estimate influence using $1{,}000$ Monte Carlo simulations and report the average cascade size. To reduce high-frequency variability due to stochastic thresholds, we apply a centered moving average with a window of nine grid points.

This analysis provides a direct view of how amplification strength translates into incremental influence, allowing us to assess whether reinforcement effects scale smoothly or exhibit threshold-like behavior.

\subsection{Seed Selection under Reinforcement}

We begin by examining how reinforcement alters the composition of optimal seed sets. Figure~\ref{fig:timeline} and Table~\ref{tab:timeline} compare the sequence of nodes selected by the CELF algorithm under the LT and PT models on the Facebook network for a budget of $k=20$.

\begin{figure*}[htbp!]
\centering
\includegraphics[width=\columnwidth,height=4.9cm]{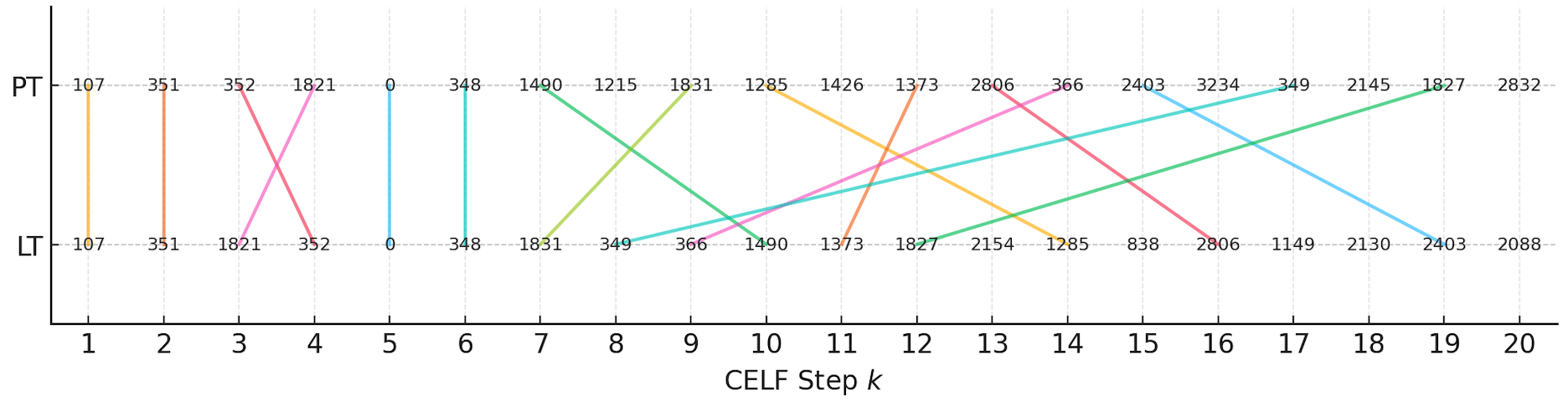}
\caption{CELF seed-selection timeline for the Facebook network (PT vs.\ LT, \(k=20\)). Horizontal position shows the selection step; vertical links mark common seeds. Early overlap is followed by divergence as PT accounts for reinforcement effects.}
\label{fig:timeline}
\end{figure*}

\begin{table}[h!]
\centering
\begin{tabular}{p{0.20\columnwidth} p{0.64\columnwidth}}
\hline\hline
Model & Seeds selected (in order) \\
\hline
Linear Threshold   & 107, 351, 1821, 352, 0, 348, 1831, 349, 366, 1490, 1373, 1827, 2154, 1285, 838, 2806, 1149, 2130, 2403, 2088 \\
\hline
Pressure Threshold & 107, 351, 352, 1821, 0, 348, 1490, 1215, 1831, 1285, 1426, 1373, 2806, 366, 2403, 3234, 349, 2145, 1827, 2832 \\
\hline
\hline
\end{tabular}\vspace{10pt}
\caption{Ordered seed sets returned by CELF for the Facebook network (\(k=20\)).}
\label{tab:timeline}
\end{table}

The two models exhibit strong agreement in the early stages of the selection process. The first several nodes selected under PT and LT largely coincide, indicating that nodes with strong intrinsic influence, those capable of generating substantial spread independently, remain important even when reinforcement is present.

However, meaningful divergence emerges in later selection steps. Beginning around iteration $8$, PT selects nodes that do not appear in the LT seed set, while omitting several nodes that are included under LT. This divergence reflects a fundamental shift in how marginal gains are evaluated under reinforcement.

Under the PT model, the marginal contribution of a node depends not only on its direct influence, but also on the extent to which it enables amplification through its neighborhood. Nodes embedded in densely connected regions are more likely to receive substantial incoming pressure before activation, which in turn increases their outgoing influence after activation. As a result, PT tends to favor nodes that participate in tightly knit local structures, where reinforcement can accumulate. This reflects the complementarity effects identified in Section~\ref{sec:decision}, where the value of a node depends on its interaction with previously selected seeds.

In contrast, the LT model evaluates nodes solely on static influence weights, without accounting for downstream amplification. This leads to a preference for nodes with high immediate reach, even if their neighborhoods do not support strong reinforcement effects.

\paragraph{Implications for intervention strategy.}
These results indicate that reinforcement can fundamentally alter optimal seeding strategies. While high-degree or high-centrality nodes remain relevant, PT shifts emphasis toward nodes that are positioned to benefit from coordinated activation and repeated exposure.

From a managerial perspective, this suggests that effective interventions in reinforcement-driven settings should prioritize not just influential individuals, but also cohesive groups in which social pressure can accumulate. Strategies based solely on static centrality measures may fail to identify nodes that are critical for triggering amplification-driven cascades.

\subsection{Influence as a Function of Budget}

We next examine how reinforcement affects realized diffusion outcomes as a function of intervention budget. Figures~\ref{fig:FB}--\ref{fig:Rand} plot the expected influence spread as the seed budget $k$ increases across the four networks under both LT and PT.

\begin{figure}[t]
    \centering
    \begin{minipage}{.48\textwidth}
    \centering
    \resizebox{\linewidth}{!}{%
        \begin{tikzpicture}
\definecolor{color0}{rgb}{0.12156862745098,0.466666666666667,0.705882352941177}
\definecolor{color1}{rgb}{1,0.498039215686275,0.0549019607843137}
\definecolor{color2}{rgb}{0.172549019607843,0.627450980392157,0.172549019607843}
\begin{axis}[
legend cell align={left},
reverse legend,
legend style={
  fill opacity=0.8,
  draw opacity=1,
  text opacity=1,
  at={(0.97,0.25)},
  anchor=east,
  draw=white!80!black
},
tick align=outside,
tick pos=left,
x grid style={white!69.0196078431373!black},
xlabel={Budget \(\displaystyle k\)},
xmajorgrids,
xmin=0, xmax=61,
xtick style={color=black},
y grid style={white!69.0196078431373!black},
ylabel style={yshift=5pt},
ylabel={Influence},
ymajorgrids,
ymin=0, ymax=3600,
ytick style={color=black}
]
\addplot [thick, color0, mark=*, mark size=.5, mark options={solid}]
table {%
1 334.642
2 608.443
3 833.666
4 1020.248
5 1181.6
6 1262.681
7 1334.467
8 1384.669
9 1430.347
10 1472.063
11 1515.548
12 1551.514
13 1589.55
14 1626.002
15 1656.765
16 1687.273
17 1713.757
18 1739.178
19 1763.244
20 1787.961
21 1812.06
22 1836.995
23 1858.174
24 1879.514
25 1901.686
26 1922.769
27 1942.921
28 1968.146
29 1989.563
30 2009.174
31 2028.747
32 2047.13
33 2066.004
34 2084
35 2101.559
36 2118.451
37 2135.81
38 2152.409
39 2169.873
40 2186.079
41 2202.106
42 2217.497
43 2232.437
44 2249.494
45 2264.231
46 2278.804
47 2292.923
48 2307.515
49 2322.161
50 2335.654
51 2348.992
52 2361.98
53 2374.705
54 2387.533
55 2400.321
56 2412.639
57 2425.011
58 2437.552
59 2449.382
60 2461.219
};
\addlegendentry{Linear Threshold ($\alpha=0$)}
\addplot [thick, color1, mark=*, mark size=.5, mark options={solid}]
table {%
1 391.72
2 688.303
3 922.964
4 1151.397
5 1312.088
6 1400.988
7 1472.899
8 1537.7
9 1592.164
10 1645.537
11 1695.249
12 1743.071
13 1790.783
14 1834.418
15 1873.344
16 1909.028
17 1946.952
18 1981.673
19 2015.471
20 2044.356
21 2073.682
22 2101.406
23 2131.614
24 2158.072
25 2184.315
26 2208.925
27 2233.526
28 2259.445
29 2284.054
30 2307.706
31 2330.674
32 2352.734
33 2373.766
34 2394.842
35 2414.687
36 2433.409
37 2452.353
38 2471.973
39 2491.125
40 2508.848
41 2525.745
42 2545.317
43 2562.675
44 2579.178
45 2597.966
46 2613.8
47 2631.358
48 2645.633
49 2661.032
50 2675.743
51 2689.133
52 2703.465
53 2717.398
54 2731.408
55 2744.001
56 2756.562
57 2769.202
58 2782.859
59 2797.137
60 2810.784
};
\addlegendentry{Pressure Threshold ($\alpha=0.001$)}
\addplot [thick, color2, mark=*, mark size=.5, mark options={solid}]
table {%
1 660.579
2 1078.835
3 1455.629
4 1727.461
5 1907.756
6 2033.591
7 2137.707
8 2222.465
9 2310.523
10 2376.827
11 2439.832
12 2497.333
13 2545.548
14 2585.35
15 2629.415
16 2666.051
17 2703.345
18 2739.729
19 2775.337
20 2810.059
21 2839.22
22 2867.473
23 2896.07
24 2922.519
25 2947.255
26 2974.572
27 2999.281
28 3023.28
29 3046.049
30 3067.716
31 3090.573
32 3110.024
33 3130.296
34 3150.115
35 3169.333
36 3190.846
37 3210.373
38 3225.988
39 3243.655
40 3258.229
41 3272.752
42 3287.554
43 3301.066
44 3314.712
45 3329.657
46 3344.561
47 3356.026
48 3368.143
49 3379.119
50 3393.093
51 3403.333
52 3413.243
53 3423.035
54 3432.982
55 3445.419
56 3455.29
57 3464.391
58 3473.042
59 3481.419
60 3490.28
};
\addlegendentry{Pressure Threshold ($\alpha = 0.005$)}
\end{axis}
\end{tikzpicture}
    }
    \caption{Influence vs.\ budget for the Facebook network.}
    \label{fig:FB}
    \end{minipage}%
    \hfill
    \begin{minipage}{.48\textwidth}
    \centering
    \resizebox{\linewidth}{!}{%
        \begin{tikzpicture}
\definecolor{color0}{rgb}{0.12156862745098,0.466666666666667,0.705882352941177}
\definecolor{color1}{rgb}{1,0.498039215686275,0.0549019607843137}
\definecolor{color2}{rgb}{0.172549019607843,0.627450980392157,0.172549019607843}
\begin{axis}[
legend cell align={left},
reverse legend,
legend style={
  fill opacity=0.8,
  draw opacity=1,
  text opacity=1,
  at={(0.97,0.25)},
  anchor=east,
  draw=white!80!black
},
tick align=outside,
tick pos=left,
x grid style={white!69.0196078431373!black},
xlabel={Budget \(\displaystyle k\)},
xmajorgrids,
xmin=0, xmax=61,
xtick style={color=black},
y grid style={white!69.0196078431373!black},
ylabel style={yshift=5pt},
ylabel={Influence},
ymajorgrids,
ymin=0, ymax=7500,
ytick style={color=black}
]
\addplot [thick, color0, mark=*, mark size=.5, mark options={solid}]
table {%
1 724.446
2 1169.028
3 1565.53
4 1921.167
5 2240.728
6 2490.144
7 2710.867
8 2883.476
9 3057.97
10 3220.825
11 3354.994
12 3481.557
13 3592.957
14 3695.593
15 3801.452
16 3891.044
17 3978.995
18 4068.885
19 4151.053
20 4228.357
21 4300.316
22 4374.764
23 4440.992
24 4502.643
25 4560.075
26 4613.92
27 4669.627
28 4730.352
29 4780.217
30 4829.166
31 4874.286
32 4916.201
33 4959.251
34 5000.799
35 5043.995
36 5080.665
37 5116.117
38 5152.02
39 5189.203
40 5221.273
41 5254.385
42 5286.392
43 5315.155
44 5344.004
45 5372.311
46 5401.982
47 5429.512
48 5458.084
49 5485.69
50 5511.501
51 5536.72
52 5560.478
53 5582.978
54 5606.244
55 5627.405
56 5648.86
57 5668.739
58 5689.415
59 5707.882
60 5725.608
};
\addlegendentry{Linear Threshold ($\alpha=0$)}
\addplot [thick, color1, mark=*, mark size=.5, mark options={solid}]
table {%
1 1054.106
2 1922.544
3 2697.091
4 3393.084
5 3958.991
6 4410.341
7 4726.157
8 5012.502
9 5279.115
10 5500.902
11 5661.572
12 5825.633
13 5987.494
14 6124.452
15 6269.464
16 6372.591
17 6466.473
18 6549.754
19 6620.768
20 6688.814
21 6760.987
22 6827.113
23 6879.977
24 6920.436
25 6957.018
26 6986.451
27 7017.48
28 7051.771
29 7082.021
30 7105.058
31 7115
32 7115
33 7115
34 7115
35 7115
36 7115
37 7115
38 7115
39 7115
40 7115
41 7115
42 7115
43 7115
44 7115
45 7115
46 7115
47 7115
48 7115
49 7115
50 7115
51 7115
52 7115
53 7115
54 7115
55 7115
56 7115
57 7115
58 7115
59 7115
60 7115
};
\addlegendentry{Pressure Threshold ($\alpha=0.001$)}
\addplot [thick, color2, mark=*, mark size=.5, mark options={solid}]
table {%
1 4773.406
2 5949.983
3 6604.243
4 6961.594
5 7115
6 7115
7 7115
8 7115
9 7115
10 7115
11 7115
12 7115
13 7115
14 7115
15 7115
16 7115
17 7115
18 7115
19 7115
20 7115
21 7115
22 7115
23 7115
24 7115
25 7115
26 7115
27 7115
28 7115
29 7115
30 7115
31 7115
32 7115
33 7115
34 7115
35 7115
36 7115
37 7115
38 7115
39 7115
40 7115
41 7115
42 7115
43 7115
44 7115
45 7115
46 7115
47 7115
48 7115
49 7115
50 7115
51 7115
52 7115
53 7115
54 7115
55 7115
56 7115
57 7115
58 7115
59 7115
60 7115
};
\addlegendentry{Pressure Threshold ($\alpha=0.005$)}
\end{axis}
\end{tikzpicture}
    }
    \caption{Influence vs.\ budget for the Wikipedia network.}
    \label{fig:Wiki}
    \end{minipage}
\end{figure}

\begin{figure}[t]
    \centering
    \begin{minipage}{.48\textwidth}
    \centering
    \resizebox{\linewidth}{!}{%
        \begin{tikzpicture}
\definecolor{color0}{rgb}{0.12156862745098,0.466666666666667,0.705882352941177}
\definecolor{color1}{rgb}{1,0.498039215686275,0.0549019607843137}
\definecolor{color2}{rgb}{0.172549019607843,0.627450980392157,0.172549019607843}
\begin{axis}[
legend cell align={left},
reverse legend,
legend style={
  fill opacity=0.8,
  draw opacity=1,
  text opacity=1,
  at={(0.97,0.25)},
  anchor=east,
  draw=white!80!black
},
tick align=outside,
tick pos=left,
x grid style={white!69.0196078431373!black},
xlabel={Budget \(\displaystyle k\)},
xmajorgrids,
xmin=0, xmax=61,
xtick style={color=black},
y grid style={white!69.0196078431373!black},
ylabel style={yshift=5pt},
ylabel={Influence},
ymajorgrids,
ymin=0, ymax=5200,
ytick style={color=black}
]
\addplot [thick, color0, mark=*, mark size=.5, mark options={solid}]
table {%
1 828.423
2 1210.393
3 1525.546
4 1798.743
5 2012.038
6 2192.477
7 2359.435
8 2510.854
9 2655.279
10 2766.19
11 2866.307
12 2963.699
13 3061.837
14 3144.634
15 3228.139
16 3303.154
17 3368.111
18 3427.513
19 3485.074
20 3538.395
21 3591.475
22 3640.215
23 3684.06
24 3728.124
25 3771.831
26 3816.506
27 3861.598
28 3898.594
29 3933.908
30 3967.226
31 3999.591
32 4032.339
33 4063.467
34 4092.953
35 4120.236
36 4147.395
37 4173.717
38 4197.457
39 4218.886
40 4240.268
41 4261.696
42 4283.131
43 4303.098
44 4323.809
45 4341.36
46 4360.538
47 4378.264
48 4394.572
49 4409.869
50 4425.074
51 4440.675
52 4455.477
53 4470.959
54 4484.948
55 4498.833
56 4512.495
57 4526.011
58 4539.067
59 4552.136
60 4565.969
};
\addlegendentry{Linear Threshold ($\alpha=0$)}
\addplot [thick, color1, mark=*, mark size=.5, mark options={solid}]
table {%
1 864.599
2 1274.755
3 1610.443
4 1890.631
5 2117.072
6 2307.742
7 2478.306
8 2631.372
9 2781.393
10 2895.744
11 3002.54
12 3102.612
13 3198.584
14 3284.893
15 3367.647
16 3440.966
17 3505.275
18 3565.901
19 3623.79
20 3674.454
21 3725.641
22 3772.414
23 3816.699
24 3860.671
25 3904.088
26 3950.613
27 3992.299
28 4027.321
29 4061.578
30 4094.259
31 4127.623
32 4158.537
33 4191.304
34 4217.246
35 4243.633
36 4268.515
37 4293.71
38 4317.123
39 4339.077
40 4360.589
41 4380.307
42 4399.322
43 4417.985
44 4434.989
45 4452.783
46 4469.303
47 4485.665
48 4503.364
49 4519.749
50 4534.898
51 4549.35
52 4564.905
53 4580.528
54 4594.44
55 4608.196
56 4622.271
57 4635.627
58 4649.094
59 4662.315
60 4673.804
};
\addlegendentry{Pressure Threshold ($\alpha=0.001$)}
\addplot [thick, color2, mark=*, mark size=.5, mark options={solid}]
table {%
1 1083.096
2 1592.946
3 2023.862
4 2344.98
5 2598.107
6 2823.061
7 3005.649
8 3191.511
9 3334.292
10 3463.046
11 3574.688
12 3677.572
13 3772.426
14 3850.951
15 3918.337
16 3981.264
17 4041.128
18 4091.113
19 4136.607
20 4180.012
21 4225.34
22 4266.772
23 4308.036
24 4348.923
25 4384.844
26 4418.224
27 4450.732
28 4481.898
29 4511.866
30 4539.837
31 4563.312
32 4589.278
33 4610.818
34 4631.291
35 4649.316
36 4666.999
37 4684.361
38 4700.332
39 4716.425
40 4732.569
41 4748.532
42 4763.632
43 4780.511
44 4794.649
45 4808.452
46 4821.548
47 4835.095
48 4849.015
49 4862.436
50 4874.262
51 4885.248
52 4896.737
53 4908.094
54 4918.627
55 4928.907
56 4938.522
57 4948.387
58 4958.293
59 4968.401
60 4977.789
};
\addlegendentry{Pressure Threshold ($\alpha=0.005$)}
\end{axis}
\end{tikzpicture}
    }
    \caption{Influence vs.\ budget for the Bitcoin network.}
    \label{fig:BC}
    \end{minipage}%
    \hfill
    \begin{minipage}{.48\textwidth}
    \centering
    \resizebox{\linewidth}{!}{%
        \begin{tikzpicture}
\definecolor{color0}{rgb}{0.12156862745098,0.466666666666667,0.705882352941177}
\definecolor{color1}{rgb}{1,0.498039215686275,0.0549019607843137}
\definecolor{color2}{rgb}{0.172549019607843,0.627450980392157,0.172549019607843}
\begin{axis}[
legend cell align={left},
reverse legend,
legend style={
  fill opacity=0.8,
  draw opacity=1,
  text opacity=1,
  at={(0.97,0.25)},
  anchor=east,
  draw=white!80!black
},
tick align=outside,
tick pos=left,
x grid style={white!69.0196078431373!black},
xlabel={Budget \(\displaystyle k\)},
xmajorgrids,
xmin=0, xmax=61,
xtick style={color=black},
y grid style={white!69.0196078431373!black},
ylabel style={yshift=5pt},
ylabel={Influence},
ymajorgrids,
ymin=0, ymax=5500,
ytick style={color=black}
]
\addplot [thick, color0, mark=*, mark size=.5, mark options={solid}]
table {%
1 236.911
2 496.616
3 702.837
4 876.178
5 1067.877
6 1232.833
7 1392.918
8 1546.086
9 1703.67
10 1847.804
11 1973.792
12 2105.853
13 2240.152
14 2368.532
15 2484.657
16 2599.259
17 2727.107
18 2824.584
19 2921.274
20 3011.234
21 3102.095
22 3214.637
23 3294.092
24 3387.032
25 3461.638
26 3546.502
27 3615.883
28 3694.263
29 3757.151
30 3834.451
31 3897.627
32 3957.677
33 4018.885
34 4084.961
35 4141.081
36 4197.002
37 4254.069
38 4317.531
39 4367.534
40 4418.623
41 4468.527
42 4516.362
43 4561.494
44 4612.05
45 4656.926
46 4701.074
47 4740.73
48 4779.572
49 4816.236
50 4852.233
51 4888.328
52 4921.855
53 4954.63
54 4988.449
55 5000.395
56 5000.381
57 5000.314
58 5000.528
59 5000.517
60 5000.673
};
\addlegendentry{Linear Threshold ($\alpha=0$)}
\addplot [thick, color1, mark=*, mark size=.5, mark options={solid}]
table {%
1 301.215
2 554.528
3 773.746
4 1007.202
5 1218.937
6 1403.755
7 1593.759
8 1789.017
9 1956.993
10 2129.5
11 2281.41
12 2438.065
13 2599.626
14 2744.491
15 2878.351
16 3012.517
17 3131.448
18 3253.815
19 3378.235
20 3508.045
21 3614.405
22 3714.713
23 3817.254
24 3919.004
25 4006.886
26 4090.31
27 4179.237
28 4259.916
29 4333.491
30 4406.899
31 4477.306
32 4549.979
33 4617.915
34 4680.458
35 4740.727
36 4800.141
37 4856.382
38 4910.448
39 4961.671
40 5000.792
41 5000.805
42 5000.75
43 5000.11
44 5000.744
45 5000.891
46 5000.681
47 5000.209
48 5000.999
49 5000.928
50 5000.607
51 5000.577
52 5000.774
53 5000.228
54 5000.876
55 5000.846
56 5000.896
57 5000.017
58 5000.106
59 5000.862
60 5000.469
};
\addlegendentry{Pressure Threshold ($\alpha=0.001$)}
\addplot [thick, color2, mark=*, mark size=.5, mark options={solid}]
table {%
1 409.414
2 767.34
3 1103.733
4 1442.252
5 1751.58
6 2033.052
7 2307.674
8 2537.749
9 2762.884
10 2985.4
11 3213.917
12 3401.161
13 3593.41
14 3775.471
15 3960.66
16 4127.566
17 4289.577
18 4432.05
19 4555.37
20 4678.688
21 4783.509
22 4887.983
23 4985.327
24 5000
25 5000
26 5000
27 5000
28 5000
29 5000
30 5000
31 5000
32 5000
33 5000
34 5000
35 5000
36 5000
37 5000
38 5000
39 5000
40 5000
41 5000
42 5000
43 5000
44 5000
45 5000
46 5000
47 5000
48 5000
49 5000
50 5000
51 5000
52 5000
53 5000
54 5000
55 5000
56 5000
57 5000
58 5000
59 5000
60 5000
};
\addlegendentry{Pressure Threshold ($\alpha=0.005$)}
\end{axis}
\end{tikzpicture}
    }
    \caption{Influence vs.\ budget for the Erd\H{o}s--R\'enyi network.}
    \label{fig:Rand}
    \end{minipage}
\end{figure}

Across all networks, the PT model produces consistently larger influence spreads than LT for the same budget. This confirms that reinforcement increases the effectiveness of seeding interventions not merely by accelerating diffusion, but by expanding the set of nodes that can ultimately be activated.

However, the magnitude of these gains varies substantially across networks, revealing an important interaction between reinforcement and network structure. In the Facebook and Wikipedia networks, which have relatively high edge-to-node ratios and pronounced clustering, PT yields significantly larger gains over LT. In contrast, the Bitcoin network, which is much sparser and less clustered, exhibits only modest improvements under PT.

This variation reflects the mechanism underlying pressure-based amplification. Reinforcement in the PT model depends on nodes receiving sufficient cumulative influence prior to activation. Dense and clustered networks provide multiple pathways for influence to accumulate, increasing the likelihood that nodes activate under high pressure and subsequently amplify their outgoing influence. In sparse networks, opportunities for such accumulation are limited, reducing the scope for reinforcement to affect diffusion.

The Erd\H{o}s--R\'enyi network provides a useful intermediate benchmark. Although random, its edge density enables reinforcement effects to accumulate more readily than in the Bitcoin network, leading to faster growth in influence under PT. In this setting, we observe that PT can reach full network activation at relatively low budgets compared to LT, illustrating how amplification can accelerate saturation in sufficiently connected graphs.

\paragraph{Shape of diffusion curves.}
Despite these differences in scale, the qualitative shapes of the influence curves remain similar across models. Both LT and PT exhibit diminishing returns as the budget increases, with early seeds generating larger marginal gains than later ones even in the absence of global submodularity under the PT model. This suggests that while reinforcement alters the magnitude of diffusion, it does not fundamentally eliminate saturation effects arising from overlapping neighborhoods and limited remaining reachable nodes.

At the same time, PT shifts the entire influence curve upward, effectively increasing the return on each seed within the regime where diffusion has not yet saturated. This upward shift is not uniform across networks, but is largest in those where structural conditions enable repeated exposure and reinforcement.

\paragraph{Implications for intervention design.}
These results highlight that the value of modeling reinforcement depends critically on network structure. In dense, highly connected environments such as online social platforms or tightly knit communities, ignoring reinforcement can lead to substantial underestimation of achievable influence and misallocation of seeding resources.

Conversely, in sparse or weakly connected networks, the incremental benefit of accounting for reinforcement may be limited, suggesting that simpler models such as LT may suffice for decision-making.

Overall, the results indicate that the effectiveness of influence interventions under PT is mediated by the interaction between seeding strategies and endogenous amplification, emphasizing the importance of aligning diffusion models with the structural characteristics of the target network.

\subsection{Sensitivity to Reinforcement Strength}

We finally examine how the strength of social reinforcement affects diffusion outcomes. Figure~\ref{fig:alpha} plots the expected influence spread under the PT model as a function of the amplification parameter $\alpha$ for the Facebook network with a fixed seed budget of $k=10$.

\begin{figure}[t]
\centering
\resizebox{.48\columnwidth}{!}{%
    \input{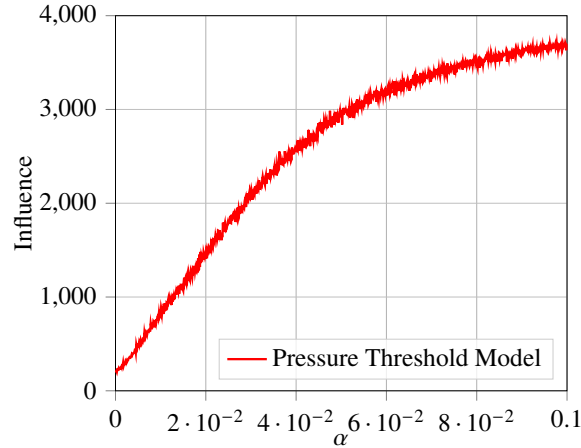}
}
\caption{Effect of the reinforcement parameter $\alpha$ on influence spread (Facebook network, $k=10$). Values are smoothed using a centered moving average (window size 9).}
\label{fig:alpha}
\end{figure}

As expected, influence spread increases with the level of reinforcement. When $\alpha = 0$, the PT model reduces to the Linear Threshold model, yielding the baseline level of diffusion. As $\alpha$ increases, nodes that activate under social pressure exert greater downstream influence, leading to larger cascades. This effect arises because higher amplification increases the likelihood that nodes activate under sufficient pressure to trigger further reinforcement downstream.

The relationship between $\alpha$ and influence, however, is not purely linear. At low values of $\alpha$, increases in reinforcement produce gradual improvements in spread, reflecting incremental amplification of local influence. In an intermediate regime, the marginal impact of $\alpha$ becomes more pronounced, as reinforcement begins to systematically enable additional nodes to cross their activation thresholds. At higher values of $\alpha$, the curve begins to flatten, indicating diminishing returns to additional amplification as the network approaches saturation.

These regimes illustrate how the effect of reinforcement depends both on its magnitude and on the underlying network structure. When amplification is weak, diffusion remains close to the LT baseline. When amplification is moderate, reinforcement meaningfully expands the reachable portion of the network. Beyond a certain point, further increases in $\alpha$ primarily accelerate diffusion rather than expand its eventual reach.

\paragraph{Interpretation for model calibration.}
From a modeling perspective, these results highlight the importance of calibrating $\alpha$ to empirical data. Extremely small values may fail to capture observed reinforcement effects, while excessively large values may overstate amplification and lead to unrealistically fast or saturated cascades.

In practical settings, $\alpha$ can be interpreted as a behavioral sensitivity parameter governing how strongly individuals internalize social pressure. Estimating this parameter from observed cascade data, as explored in Section~\ref{sec:validation}, is therefore critical for accurate prediction and reliable decision support.

\paragraph{Implications for intervention design.}
For decision makers, the results suggest that the impact of reinforcement-aware strategies depends on the regime in which the system operates. In environments with moderate reinforcement, accounting for amplification can substantially increase the effectiveness of seeding interventions. In contrast, when reinforcement is either negligible or already saturated, the incremental benefit of adjusting strategies may be limited.

Overall, this analysis shows that reinforcement is not simply a binary feature of diffusion, but a tunable mechanism whose strength shapes both the scale and efficiency of influence propagation. More broadly, these findings highlight that reinforcement is not merely an incremental extension of classical diffusion but a mechanism that reshapes both the scale and structure of influence propagation, with direct implications for intervention design.

\vspace{.1in} \section{Empirical Evaluation of Approximation Guarantees}
\label{sec:greedy}

While Section~\ref{sec:model_properties} establishes theoretical bounds for influence maximization under the Pressure Threshold (PT) model, it is important to understand how these guarantees translate into practice. In this section, we empirically evaluate the sandwich approximation strategy on a real-world network and quantify the resulting data-dependent performance guarantees. In particular, we assess whether the theoretical relaxation introduced by the sandwich framework leads to practically useful guarantees when applied to realistic network structures.

\subsection{Experimental Setup}

To assess the effectiveness of the sandwich approximation framework, we conduct experiments on a high-activity subgraph derived from the Higgs Twitter dataset. This subgraph captures realistic structural properties while remaining computationally tractable for detailed analysis. The resulting graph contains $1{,}198$ nodes and $1{,}319$ edges, with an edge-to-node ratio of $1.101$.

We set the amplification parameter to $\alpha = 0.5$, representing a regime with substantial reinforcement, and fix the seed budget at $k = 5$. This choice places the system in a high-reinforcement regime where the gap between PT and its submodular surrogates is most pronounced. For each model, Linear Threshold (LT), Pressure Threshold (PT), and the upper-bound model ($M^+$), we apply the CELF algorithm to obtain corresponding seed sets $S_{LT}$, $S_{PT}$, and $S_{M^+}$. 

Each influence estimate is computed using $1{,}000$ Monte Carlo simulations with independently sampled thresholds. Following the sandwich approximation procedure, we evaluate all candidate seed sets under the PT model and select the seed set that achieves the largest expected influence.


\subsection{Results}

\begin{table}[h!]
    \centering
    \begin{tabular}{p{0.15\columnwidth} p{0.45\columnwidth} p{0.30\columnwidth}}
    \hline\hline
    Model & Seeds selected (in order) & Expected Spread under PT \\
    \hline
    LT   & 64911, 103447, 28951, 35729, 112673 & 245.493\\
    \hline
    PT & 35729, 112673, 142899, 232429, 339158 & 258.912\\
    \hline
    $M^+$ & 1276, 112673, 142899, 232429, 339158 & 231.674\\
    \hline
    \hline
    \end{tabular}
    \vspace{10pt}
    \caption{Seed sets obtained under LT, PT, and $M^+$ along with their expected influence under PT ($k=5$, $\alpha=0.5$).}
    \label{tab:sandwich-seeds}
\end{table}

The relative ordering of seed sets reflects the differing objectives optimized under each model, with PT capturing reinforcement-aware interactions that are absent in LT and $M^+$.

Table~\ref{tab:sandwich-seeds} reports the seed sets obtained under each model and their corresponding influence evaluated under PT. Among the candidate sets, the PT-derived seed set $S_{PT}$ achieves the highest expected spread, and is therefore selected as the final solution $S_{\text{sand}}$.

While this selection reflects empirical performance, the sandwich framework also enables a formal approximation guarantee through the upper-bound model $M^+$. Specifically, we compute the ratio between the PT and $M^+$ influence values for the seed set $S_{M^+}$:
\[
\frac{\sigma_{PT}(S_{M^+})}{\sigma_{M^+}(S_{M^+})} 
= \frac{231.67}{260.00} \approx 0.891.
\]

This ratio serves as a data-dependent relaxation factor. Intuitively, it measures how tightly the upper-bound model $M^+$ approximates PT for its own optimized seed set. Combining this with the classical $(1 - 1/e)$ guarantee for submodular maximization yields the bound:
\[
\sigma_{PT}(S_{\text{sand}}) 
\;\geq\; 0.891 \cdot (1 - 1/e) \cdot \sigma_{PT}(S_{PT}^*) 
\;\approx\; 0.563 \cdot \sigma_{PT}(S_{PT}^*).
\]

This guarantee holds without requiring monotonicity or submodularity of the PT influence function.

\paragraph{Interpretation.}
This result demonstrates that, despite the loss of monotonicity and submodularity, the sandwich approximation strategy yields strong performance guarantees in practice. The achieved bound of $56.3\%$ of the optimal influence is close to the classical $63.2\%$ guarantee for submodular influence maximization, indicating that the practical loss from reinforcement-induced non-submodularity is limited in this instance. From a decision-making perspective, this suggests that reinforcement-aware strategies can be deployed with performance guarantees that remain close to classical benchmarks.

More broadly, these findings suggest that worst-case theoretical limitations may overstate the practical difficulty of influence maximization under reinforcement. By combining tractable surrogate models with data-driven evaluation, the sandwich framework provides a principled and implementable approach for optimizing interventions in reinforcement-driven diffusion settings.

This empirical evaluation of the sandwich framework bridges the gap between theoretical guarantees and practical applicability, demonstrating that reinforcement-aware optimization can be both principled and effective in real-world settings.

\vspace{.1in} \section{Empirical Validation}
\label{sec:validation}

While the previous sections evaluate the behavioral and algorithmic implications of reinforcement, we now assess whether the Pressure Threshold (PT) model provides a more accurate description of real-world diffusion. To this end, we compare PT and the Linear Threshold (LT) model against observed social media cascades and evaluate their predictive performance. This comparison allows us to assess whether incorporating reinforcement leads to systematically improved alignment between modeled and observed diffusion.

\subsection{Data and Experimental Setup}

We conduct our analysis on the Higgs Twitter dataset \cite{DeDomenico2013}, which captures information diffusion surrounding the announcement of the Higgs boson discovery in July 2012. The dataset includes multiple interaction networks: retweets, replies, mentions, and follower relationships, and provides precise timestamps that allow for temporal validation of diffusion models.

We focus on the retweet network, where a user is considered active once they retweet a message. For each experiment, we extract a cascade by selecting a seed node and collecting all users activated through subsequent retweets. This produces a sequence of real-world activations over time.

To ensure that models are evaluated conservatively, we augment each cascade subgraph by adding 500 neighboring nodes from the global follower network that were not part of the observed cascade. This penalizes models that over-predict diffusion and ensures that predictive accuracy reflects both under- and over-estimation.

We construct 100 such subgraphs, each representing a distinct cascade. For each subgraph, we record structural properties including density, clustering, and seed centrality, to support the interpretation of model performance across different network environments. These variations allow us to examine whether reinforcement improves predictive accuracy uniformly or primarily in structurally dense and highly interactive settings.

\subsection{Simulation Protocol}

For each cascade subgraph, we simulate diffusion under both PT and LT using the observed seed node as the initial activation. Edge weights and thresholds follow the same conventions used in earlier sections.

To align model dynamics with real-world timing, we discretize time such that one simulation step corresponds to 10 hours of real-world activity. For each model and parameter setting, we estimate expected influence trajectories using $1{,}000$ independent Monte Carlo simulations. This ensures that differences in predictive performance reflect model structure rather than sampling variability.

To evaluate the role of reinforcement strength, we vary the amplification parameter over the range $\alpha \in [0,1.5]$ with step size $0.05$. For each value of $\alpha$, we compute the prediction error across all sampled cascades. This enables identification of the reinforcement regime that best aligns with observed diffusion behavior.

\subsection{Evaluation Metrics}

We measure predictive accuracy using the Normalized Mean Squared Error (NMSE) over time. For a cascade with $N$ nodes, the loss at time $t$ is given by:
\[
\text{NMSE}(t) = \left( \frac{\sigma_{\text{Actual}}(t_{\text{Actual}}) - \mathbb{E}[\sigma_{\text{Predicted}}(t)]}{N} \right)^2.
\]

This normalization facilitates comparison across cascades of different sizes.

Here, $\sigma_{\text{Actual}}$ denotes the cumulative number of activated nodes observed in the real cascade, and $\mathbb{E}[\sigma_{\text{Predicted}}(t)]$ is the expected number of activations in the simulated diffusion process.

We report the final prediction error averaged across all cascades for each value of $\alpha$.

\subsection{Results}

\begin{figure}[h!]
\centering
\resizebox{.48\columnwidth}{!}{%
    \begin{tikzpicture}
\definecolor{color0}{rgb}{0.12156862745098,0.466666666666667,0.705882352941177}
\definecolor{color1}{rgb}{1,0.498039215686275,0.0549019607843137}
\definecolor{color2}{rgb}{0.172549019607843,0.627450980392157,0.172549019607843}
\definecolor{color3}{rgb}{0.32,0.135,0.32}
\begin{semilogyaxis}[
legend cell align={left},
reverse legend,
legend style={
  fill opacity=0.8,
  draw opacity=1,
  text opacity=1,
  at={(0.97,0.25)},
  anchor=east,
  draw=white!80!black
},
tick align=outside,
tick pos=left,
x grid style={white!69.0196078431373!black},
xlabel={Time \(\displaystyle t\)},
xmajorgrids,
xmin=0, xmax=20,
xtick style={color=black},
y grid style={white!69.0196078431373!black},
ylabel={Influence},
ymajorgrids,
ymin=0, ymax=300,
ytick style={color=black}
]
\addplot [thick, color0, mark=*, mark size=.5, mark options={solid}]
table {%
1 1   
2 2   
3 2   
4 5  
5 19  
6 42  
7 49  
8 68  
9 87 
10 179 
11 208 
12 217 
13 228 
14 230 
15 243 
16 246 
17 248 
18 248
19 248 
20 248 
21 248 
22 248 
23 248 
24 248 
25 248 
26 248 
27 248
};
\addlegendentry{Empirical Data}
\addplot [thick, color1, mark=*, mark size=.5, mark options={solid}]
table {%
1 1.     
2 5.159 
3 16.584 
4 24.369 
5 34.403 
6 40.183 
7 44.988 
8 48.926 
9 51.79
10 54.017 
11 55.686 
12 56.809 
13 57.679 
14 58.238 
15 58.582 
16 58.801 
17 58.921 
18 58.975 
19 58.989
20 58.992
};
\addlegendentry{Linear Threshold ($\alpha=0$)}
\addplot [thick, color2, mark=*, mark size=.5, mark options={solid}]
table {%
1 1.      
2 5.166  
3 17.515  
4 28.308  
5 48.044  
6 67.203  
7 90.074 
8 116.533
9 144.581 
10 169.924 
11 188.315 
12 201.198 
13 208.501 
14 212.734 
15 214.764 
16 215.688 
17 216.079
18 216.157 
19 216.225 
20 216.261 
21 216.29  
22 216.296 
23 216.296 
24 216.296 
25 216.296 
26 216.296
};
\addlegendentry{Pressure Threshold ($\alpha=0.2$)}
\addplot [thick, color3, mark=*, mark size=.5, mark options={solid}]
table {%
1 1.      
2 5.141  
3 18.608  
4 33.758  
5 73.369 
6 124.582 
7 183.183 
8 227.556
9 245.352 
10 247.31  
11 247.464 
12 247.466 
13 247.466 
14 247.466 
15 247.466 
16 247.466 
17 247.466
18 247.466 
19 247.466 
20 247.466
};
\addlegendentry{Pressure Threshold ($\alpha=0.5$)}
\end{semilogyaxis}
\end{tikzpicture}
}
\resizebox{.48\columnwidth}{!}{%
    \begin{tikzpicture}
\definecolor{color0}{rgb}{0.12156862745098,0.466666666666667,0.705882352941177}
\definecolor{color1}{rgb}{1,0.498039215686275,0.0549019607843137}
\definecolor{color2}{rgb}{0.172549019607843,0.627450980392157,0.172549019607843}
\definecolor{color3}{rgb}{0.32,0.135,0.32}

\begin{semilogyaxis}[
legend cell align={left},
reverse legend,
legend style={
  fill opacity=0.8,
  draw opacity=1,
  text opacity=1,
  at={(0.97,0.25)},
  anchor=east,
  draw=white!80!black
},
tick align=outside,
tick pos=left,
x grid style={white!69.0196078431373!black},
xlabel={Time \(\displaystyle t\)},
xmajorgrids,
xmin=0, xmax=20,
xtick style={color=black},
y grid style={white!69.0196078431373!black},
ylabel={Influence},
ymajorgrids,
ymin=0, ymax=300,
ytick style={color=black}
]
\addplot [thick, color0, mark=*, mark size=.5, mark options={solid}]
table {%
1 1   
2 2   
3 3  
4 11  
5 18  
6 49  
7 55  
8 62  
9 85 
10 183 
11 202 
12 212 
13 224 
14 231 
15 234 
16 242 
17 246 
18 250
19 250
20 250
21 250
};
\addlegendentry{Empirical Data}
\addplot [thick, color1, mark=*, mark size=.5, mark options={solid}]
table {%
1 1.     
2 1.     
3 8.    
4 11.031 
5 12.27  
6 12.347 
7 12.446 
8 12.525 
9 12.571 
10 12.609
11 12.663 
12 12.709 
13 12.749 
14 12.77  
15 12.786 
16 12.786 
17 12.786 
18 12.786
19 12.786
20 12.786
21 12.786
};
\addlegendentry{Linear Threshold ($\alpha=0$)}
\addplot [thick, color2, mark=*, mark size=.5, mark options={solid}]
table {%
1 1.     
2 1.     
3 8.    
4 11.677 
5 13.954 
6 14.788 
7 16.8   
8 19.595 
9 22.938 
10 26.94
11 30.285 
12 32.612 
13 34.541 
14 35.886 
15 36.782 
16 37.375 
17 37.802 
18 38.079 
19 38.203 
20 38.247
21 38.249
};
\addlegendentry{Pressure Threshold ($\alpha=0.2$)}
\addplot [thick, color3, mark=*, mark size=.5, mark options={solid}]
table {%
1 1.      
2 1.      
3 8.     
4 12.51  
5 16.742  
6 22.458  
7 43.313  
8 80.304 
9 117.46
10 146.735 
11 168.006 
12 180.146 
13 185.24  
14 187.242 
15 187.995 
16 188.156 
17 188.177 
18 188.189
19 188.189
20 188.189
21 188.189
};
\addlegendentry{Pressure Threshold ($\alpha=0.5$)}
\end{semilogyaxis}
\end{tikzpicture}
}
\caption{Predicted and observed influence over time for representative cascades.}
\label{fig:Cascades}
\end{figure}
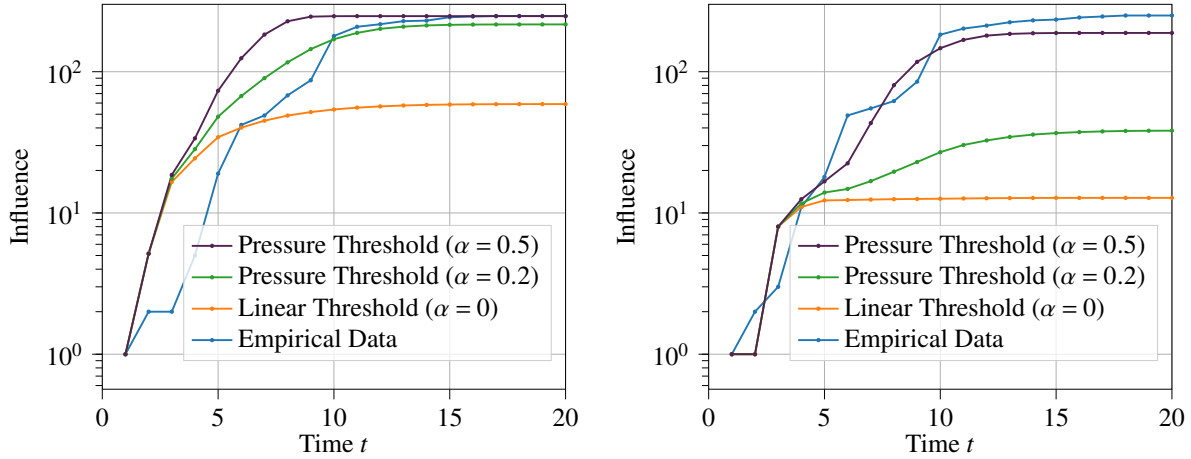

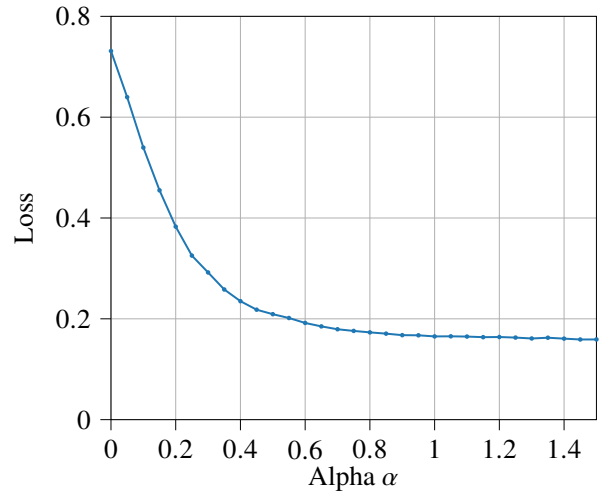
\begin{figure}[h!]
\centering
\resizebox{0.48\columnwidth}{!}{%
    \begin{tikzpicture}
\definecolor{color0}{rgb}{0.12156862745098,0.466666666666667,0.705882352941177}
\definecolor{color1}{rgb}{1,0.498039215686275,0.0549019607843137}
\definecolor{color2}{rgb}{0.172549019607843,0.627450980392157,0.172549019607843}
\begin{axis}[
legend cell align={left},
reverse legend,
legend style={
  fill opacity=0.8,
  draw opacity=1,
  text opacity=1,
  at={(0.97,0.25)},
  anchor=east,
  draw=white!80!black
},
tick align=outside,
tick pos=left,
x grid style={white!69.0196078431373!black},
xlabel={Alpha \(\alpha\)},
xmajorgrids,
xmin=0, xmax=1.5,
xtick style={color=black},
y grid style={white!69.0196078431373!black},
ylabel={Loss},
ymajorgrids,
ymin=0, ymax=0.8,
ytick style={color=black}
]
\addplot [thick, color0, mark=*, mark size=.5, mark options={solid}]
table {%
0 0.7313
0.05 0.6397
0.1 0.5397
0.15 0.4549
0.2 0.3829
0.25 0.3254
0.3 0.2920
0.35 0.2582
0.4 0.2350
0.45 0.2181
0.5 0.2092
0.55 0.2016
0.6 0.1918
0.65 0.1850
0.7 0.1793
0.75 0.1760
0.8 0.1731
0.85 0.1706
0.9 0.1676
0.95 0.1673
1 0.1651
1.05 0.1653
1.1 0.1648
1.15 0.1636
1.2 0.1639
1.25 0.1628
1.3 0.1610
1.35 0.1625
1.4 0.1607
1.45 0.1590
1.5 0.1591
};
\end{axis}
\end{tikzpicture}}
\caption{Average final NMSE as a function of the reinforcement parameter $\alpha$.}
\label{fig:EV}
\end{figure}

Figure~\ref{fig:Cascades} compares predicted and observed activation trajectories for representative cascades. The PT model consistently tracks the empirical diffusion curve more closely than LT, particularly in later stages of the cascade. This improvement is consistent with PT’s ability to capture amplification arising from repeated exposure over time. In contrast, LT systematically underestimates cascade growth due to its inability to capture amplification arising from repeated exposure.

Figure~\ref{fig:EV} plots the average prediction error across cascades as a function of $\alpha$. When $\alpha = 0$, the PT model reduces to LT and exhibits the same level of error. As $\alpha$ increases, prediction accuracy improves significantly, indicating that reinforcement plays an important role in explaining observed diffusion behavior. Importantly, accuracy improves monotonically only up to a point, after which additional amplification yields diminishing predictive returns.

The error curve stabilizes for larger values of $\alpha$, suggesting diminishing returns from additional amplification and identifying a regime in which the model is well calibrated to the data.

\paragraph{Interpretation.}
These results provide strong empirical support for the relevance of pressure-based reinforcement in real-world diffusion. By allowing influence strength to depend on social pressure at activation, the PT model captures dynamics that are not explained by static threshold models.

From a modeling perspective, the findings highlight the importance of incorporating reinforcement effects when analyzing diffusion in settings characterized by repeated exposure and social validation. From a decision-making perspective, improved predictive accuracy translates directly into better-informed intervention strategies, particularly in high-activity social systems where amplification is prominent.

These improvements in predictive accuracy directly enhance the reliability of reinforcement-aware intervention strategies derived from the PT model.

\vspace{.1in} \section{Managerial and Societal Implications}
\label{sec:implications}

Our findings have direct implications for organizations and policymakers that rely on network-based interventions. In many real-world settings, such as marketing campaigns, political communication, and online information diffusion, individual behavior is shaped not only by exposure but also by the intensity and consistency of social reinforcement. The Pressure Threshold (PT) model provides a framework for incorporating this mechanism into both prediction and decision-making. In particular, it highlights how accounting for reinforcement can lead to systematically different intervention strategies compared to classical diffusion models.

\paragraph{Implications for intervention design.}
A central insight of our analysis is that reinforcement fundamentally alters which individuals should be targeted in influence campaigns. While traditional strategies emphasize high-degree or high-centrality nodes, our results show that under reinforcement, optimal interventions prioritize nodes embedded in structurally cohesive regions of the network. These regions enable repeated exposure, increasing the likelihood of high-pressure activation and subsequent amplification. This behavior directly reflects the complementarity effects identified in Section~\ref{sec:decision}.

For practitioners, this implies that effective interventions should balance reach with reinforcement potential. Rather than focusing solely on globally influential individuals, it may be more effective to target clusters or communities where coordinated activation can generate compounding effects. This shift has implications for marketing strategies, viral content design, and grassroots mobilization efforts. In practice, this may require coordinated targeting strategies rather than isolated seeding decisions.

\paragraph{Implications for resource allocation.}
Our experiments demonstrate that the benefits of accounting for reinforcement depend strongly on network structure. In dense and highly clustered networks, reinforcement substantially increases diffusion outcomes, suggesting that more sophisticated models such as PT can significantly improve the return on investment for seeding strategies. In contrast, in sparse or weakly connected networks, the gains from reinforcement are more limited, and simpler models such as LT may suffice.

This suggests a context-dependent approach to model selection: organizations should evaluate the structural properties of the target network before committing to more complex intervention strategies. In environments where repeated exposure is common, neglecting reinforcement may lead to systematic underestimation of achievable influence.

\paragraph{Implications for platform design and governance.}
Beyond strategic interventions, our results also shed light on broader societal dynamics. Reinforcement mechanisms contribute to phenomena such as echo chambers, rapid information cascades, and the amplification of misinformation. By capturing how local pressure translates into increased influence, the PT model provides a tool for analyzing how platform design choices, such as recommendation algorithms or content exposure patterns, can shape aggregate diffusion outcomes by influencing both exposure patterns and the accumulation of social pressure.

From a policy perspective, incorporating reinforcement into diffusion models enables more realistic evaluation of interventions aimed at promoting beneficial behaviors or mitigating harmful ones. This includes applications in public health campaigns, information integrity efforts, and the regulation of online platforms.

\paragraph{Summary.}
Overall, our results highlight that influence in social networks is not solely a function of exposure, but also of reinforcement-driven conviction. Accounting for this mechanism can improve both the effectiveness of targeted interventions and the accuracy of system-level predictions in networked environments. This perspective provides a unifying lens for understanding how local interactions scale into system-level diffusion outcomes.

\vspace{.1in} \section{Conclusion}
\label{sec:conclusion}

\subsection{Summary of Contributions}

This paper studies how social reinforcement shapes diffusion dynamics and the design of network-based interventions. We introduced the Pressure Threshold (PT) model, a diffusion framework that extends the Linear Threshold (LT) model by allowing a node’s outgoing influence to depend on the social pressure that triggered its activation. This mechanism captures the formation of conviction under peer influence and the resulting amplification of downstream persuasion.

We show that incorporating reinforcement fundamentally alters the structure of influence maximization. While the problem remains NP-hard, the PT influence function is neither monotone nor submodular when $\alpha > 0$, implying that classical approximation guarantees do not directly apply. This highlights a fundamental shift from separable to interaction-driven optimization. To address this challenge, we develop a sandwich approximation framework that bounds the PT objective between tractable submodular surrogates, enabling greedy algorithms with data-dependent performance guarantees.

Our computational experiments demonstrate that reinforcement yields qualitatively different seeding strategies and substantially larger diffusion outcomes relative to LT. In particular, we show that PT favors nodes embedded in structurally cohesive regions, where repeated exposure enables amplification, rather than simply high-degree or high-centrality nodes. We further show that the benefits of reinforcement depend strongly on network structure, with significantly larger gains in dense and highly clustered networks and more modest improvements in sparse settings. These findings consistently illustrate the behavioral and structural consequences of reinforcement across both synthetic and real-world settings.

We complement these findings with an empirical evaluation of the sandwich approximation strategy, demonstrating that strong data-dependent guarantees can be achieved in practice despite the loss of classical structure. On a real-world network, the resulting solution achieves a guarantee of $56\%$ of the optimal influence, approaching the performance levels associated with submodular objectives.

Finally, using time-stamped diffusion traces from the Higgs dataset, we provide empirical validation that the PT model improves predictive accuracy relative to LT. These results confirm that pressure-based amplification is not only a theoretical construct but also a practically relevant mechanism for modeling real-world diffusion with measurable impact on both predictive accuracy and intervention effectiveness.

\subsection{Future Research}

Several directions for future research emerge from this work. On the theoretical side, an important question is to characterize regimes in which the PT influence function exhibits approximate monotonicity or submodularity, enabling stronger guarantees for greedy optimization. Identifying structural or parameter-dependent conditions under which classical 
approximation guarantees can be recovered remains an open challenge. Such characterization would clarify when reinforcement preserves enough structure to support efficient optimization.

On the empirical side, estimating the reinforcement parameter $\alpha$ from large-scale diffusion data is a key direction. Such calibration would enable the PT model to serve as a predictive tool tailored to specific domains and platforms and would enable reinforcement-aware models to be deployed as predictive tools in operational settings. Extending the model to incorporate heterogeneous sensitivity to social pressure, temporal dynamics, or evolving network structure also represents a promising avenue for future work.

Overall, our results highlight that influence in social networks is not only a function of exposure, but also of conviction formed through reinforcement. Accounting for this mechanism is essential for both understanding and effectively shaping diffusion in modern networked systems.

\vspace{.1in}
\section*{Ethical Statement}
\label{sec:ethical_statement}

The Pressure Threshold (PT) model provides a framework for understanding and optimizing information diffusion in social networks. However, such methods also have the potential for misuse in influence operations, targeted persuasion, or the amplification of misinformation.

We therefore emphasize the dual-use nature of reinforcement-aware diffusion models. When applied in real-world settings, particularly in public communication, marketing, or platform design, intervention strategies based on amplification and seed selection should be evaluated carefully for unintended consequences, including manipulation, polarization, and inequitable impact across user groups. Mechanisms for transparency, auditing, and fairness-aware design are important safeguards for responsible deployment.

All data used in this study are publicly available and anonymized. No personally identifiable information was collected or processed. 

We encourage future work to develop reinforcement-aware diffusion models that incorporate fairness, accountability, and transparency by design.

\vspace{.1in}

\section*{Acknowledgments} \label{sec:acknowledgments}
A preliminary version of this paper 
appeared at the
20th International AAAI Conference on Web and Social Media (ICWSM), 2026, held in Los Angeles, CA \cite{stutsman2026pressure}. The present manuscript substantially extends and revises that earlier version.


%
%
%


\vspace{.1in}


\bibliographystyle{informs2014} 
\bibliography{refs} 






  



\end{document}